\documentclass{commat}

\usepackage[english]{babel}
\usepackage{amsthm}
\usepackage{listings}
\usepackage{color}	
\usepackage{amsthm}
\usepackage{setspace}
\usepackage{dsfont}
\usepackage{amssymb, amsmath, bbding, pxfonts}
\usepackage{mathrsfs}
\usepackage{pifont}
\usepackage{caption}
\usepackage{subcaption}
\usepackage{comment}
\usepackage{fancyvrb}
\usepackage{float}
\usepackage{titlesec}
\usepackage[ruled,lined]{algorithm2e}
\usepackage{tabularx}
\usepackage{blindtext}
\usepackage{graphicx}
\usepackage{mwe}
\usepackage{microtype} 			
\usepackage{setspace}
\usepackage{lipsum}			
\usepackage{microtype}
\usepackage{enumerate}
\usepackage{url}
\usepackage{multicol}
\usepackage{hyperref}
\usepackage{hhline}
\usepackage{mathrsfs}
\usepackage{mathtools, nccmath}

\DeclareMathOperator{\vol}{vol}

\newcommand{\Z}{\ensuremath{\mathbb{Z}}}
\newcommand{\N}{\ensuremath{\mathbb{N}}}
\newcommand{\R}{\ensuremath{\mathbb{R}}}

\title{%
    On lattice constructions D and D' from $q$-ary linear codes
    }

\author{%
    Franciele do Carmo Silva, Ana Paula de Souza, Eleonesio Strey and Sueli Irene Rodrigues Costa
    }

\authorinfo[%
    F. C. Silva]{
    University of Campinas, Brazil}{%
    francielecs@ime.unicamp.br
    }

\authorinfo[%
    A. P. Souza]{
    University of Campinas, Brazil}{%
    anasouza@ime.unicamp.br
    }

\authorinfo[%
    E. Strey]{
    Federal University of Espírito Santo, Brazil}{%
    eleonesio.strey@ufes.br
    }

\authorinfo[%
    S. I. R. Costa]{
    University of Campinas, Brazil}{%
    sueli@ime.unicamp.br
    }

\abstract{%
    Multilevel lattice codes, such as those associated with Constructions $C$, $\overline{D}$, D and D', have relevant applications in communications. In this paper, we investigate some properties of lattices obtained via Constructions D and D' from $q$-ary linear codes. Connections with Construction A, generator matrices, expressions and bounds for the lattice volume and minimum distances are derived. Extensions of previous results regarding construction and decoding of binary and $p$-ary  linear codes ($p$ prime) are also presented.
    }



\begin{document}
\VOLUME{31}
\YEAR{2023}
\NUMBER{2}
\firstpage{173}
\DOI{https://doi.org/10.46298/cm.11146}



\section{Introduction}

Lattices are discrete additive subgroups of $\R^{n}$ that have attracted attention, due to several applications in coding for reliable and secure communications. Through their rich algebraic and geometric structures, they can achieve the capacity of the additive white Gaussian channel (AWGN) \cite{erez2004achieving}. Regarding security, lattices have been also used in coding for wiretap channels \cite{damir2021well} and currently compose one of the main approaches in the so-called Post-Quantum Cryptography \cite{peikert2016decade}.

The association of lattices with codes is natural \cite{conway2013sphere}, however, lattice code construction with good performance and practical decoding is still a hard problem. In order to reduce the decoding complexity, a possible direction is the construction of multilevel lattices from a family of nested codes, which allows multistage decoding. Similar techniques are also applied in a more general sense, as introduced in \cite{imai1977new}, to obtain multilevel lattice codes, even when the constructions do not necessarily form a lattice, what include the so-called Constructions $\overline{D}$ \cite{forney2000sphere}, C \cite{leech1971sphere}, $C^{\ast}$ \cite{bollauf2019multilevel}, D and D' \cite{leech1964some, bos1982further, barnes1983new, conway2013sphere, forney2000sphere}. These constructions are extensively studied, especially for the binary case, and appear in papers such as \cite{sadeghi2010, silva2020multilevel, bollauf2019multilevel, zhou2021encoding} and references therein. Recent works deal with generalizations of Constructions D, D' and $\overline{\mbox{D}}$ to linear codes over finite fields \cite{feng2011lattice}, codes over the ring $\Z_q$ of integers modulo $q$ \cite{strey2017lattices, strey2018bounds} and for cyclic codes over finite fields (Construction $D^{(cyc)}$) \cite{hu2020strongly}. It is well-known that some remarkable lattices with higher coding gain can be described via binary code Constructions D  and D', as turbo lattices \cite{sakzad2010construction}, the Barnes-Wall lattices \cite{conway2013sphere} and LDPC lattices \cite{sadeghi2006low}. Several proposals and analyses of multistage decoding have been presented in \cite{silva2020multilevel, sadeghi2006low, zhou2022construction, matsumine2018construction}.

Regarding codes over finite rings, a great interest came from the discovery of good nonlinear binary codes connected via the Gray map to linear codes over $\Z_{4}$ \cite{calderbank1994z4}. This study motivated several works to consider codes over more general finite rings, such as $\Z_{2k}$ and $\Z_{2^{k}}$, and their respective Gray maps \cite{wan1997quaternary, bannai1999type, dougherty2011codes}. In particular, self-dual codes over $\Z_{2k}$ have attracted interest because of their connection with even unimodular lattices \cite{bonnecaze1995quaternary, dougherty2011codes}. Under these motivations, in this paper, we focus on Constructions D, D' and A from nested linear codes over $\Z_q$. Our objective is to study some general properties of these constructions, such as volume, $L_{\mathrm{P}}$-minimum distance, with $1 \leq \mathrm{P} \leq \infty$, and bounds for coding gain. For this, we establish some relations between Construction D' and A and present bounds for these parameters in terms of their underlying codes or their duals. We also extend a multistage decoding method with re-encoding to Construction D' from $q$-ary linear codes under specific conditions.

This paper is organized as follows. Concepts and preliminary results are presented in Section 2. In Section 3, it is pointed out some known properties of Constructions D and D' and by the association of Construction D' with Construction A (Corollary \ref{teomatrizgeradoraD'}), expressions for a generator matrix (Corollary \ref{cormatrizgeradoraD'} and Corollary \ref{CorGeneralizaZhou}), volume (Corollary \ref{CorVolumeDlinhaviaA} and Remark \ref{volviaConstA}) and minimum distance (Corollary \ref{cormindistance}) of this construction are derived. In Section 4, we obtain a lower and an upper bound, respectively, for the volume of the lattices obtained by Constructions D and D' (Theorem \ref{num_max_pontos}, \ref{num_max_pontosDlinha}) and discuss specific conditions such that they can be achieved (Theorem \ref{T2} and Corollaries \ref{ConstDlinhaLI} and \ref{CorVolumeDlinhaCondTeorema}). Also, it is characterized by the $L_{\mathrm{P}}$-minimum distance and coding gain of lattices obtained via these constructions under certain conditions by using the minimum distance of the nested codes or their duals (Theorems \ref{TeoDistLpDbarra}, Corollaries \ref{cordistLpConstD}, \ref{distLpDlinhaigualdade}, \ref{PropCodinggain}). Specific minimum distance bounds for lattices from binary codes are derived (Theorem \ref{TeodistLpDlinha}). In Section 5, a known multistage decoding method \cite{zhou2022construction} with re-encoding for Construction D' over binary codes is extended to $q$-ary codes under specific conditions. Concluding remarks are included in Section \ref{SecConclusao}.

\section{Preliminaries}

This section is devoted to presenting some concepts, notations and results to be used in the next sections. We may quote \cite{conway2013sphere} and \cite{zamir2014lattice} as general references.

Our notations follow the convention for vectors in $\R^{n}$, as well as $n$-tuples in $\Z_{q}^{n}$, in bold letters and $\boldsymbol{0}$ denotes the null vector. The mapping $\rho: \Z \rightarrow \Z_{q}$ is the natural reduction ring homomorphism and $\sigma: \Z_q \rightarrow \Z$ is the standard inclusion map, extended to vectors and matrices in a component-wise way. For simplicity, we abuse the notation, using $\sigma$ and $\rho$ for $\Z_q$ and $\Z_{q}^{n}$ and omitting them in the numerical examples. When these maps are associated with $\Z_{q^{a}}$ or $\Z_{q^{a}}^{n}$, with $a > 1$, we will refer to them as $\sigma_{q^{a}}$ and $\rho_{q^{a}}$, respectively. The order of an element $\boldsymbol{h} \in \Z_q^{n}$, denoted by $\mathcal{O}(\boldsymbol{h})$, is defined as the smallest natural $m$ such that $m\boldsymbol{h} = \boldsymbol{0}$ in $\Z_{q}^{n}$ ( i.e., $m \sigma(\boldsymbol{h}) \equiv \boldsymbol{0} \mod q$).

A $q$-ary linear code $\mathcal{C}$ of length $n$ over $\Z_q$ is a $\Z_q$-module of $\Z_{q}^{n}$, that is, an additive subgroup of $\Z_{q}^{n}$. The terminology of $q$-ary codes is also applied in the study of codes over finite fields $\mathbb{F}_{q}$, however, in this work we use $q$-ary code to refer to a code over $\Z_{q}$. The code generated by the $n$-tuples $\boldsymbol{b}_1, \ldots, \boldsymbol{b}_k \in \Z_q^{n}$ is denoted by $\mathcal{C} = \langle \boldsymbol{b}_1, \ldots, \boldsymbol{b}_k\rangle$. We say that a set $\left\{\boldsymbol{b}_1, \ldots, \boldsymbol{b}_k\right\}$ is a basis for $\mathcal{C}$ if they are linearly independent over $\Z_q$ and they generate $\mathcal{C}$. In contrast to codes over finite fields, when $q$ is not a prime number there are $q$-ary linear codes that do not admit a basis. Despite this, every $q$-ary code $\mathcal{C}$ can be characterized by a minimal set of generators, due to its finitely generated module structure \cite{rotman2015advanced, jorge2012reticulados}. For a code $\mathcal{C}$, two different minimal sets of generators always have the same cardinality \cite{rotman2015advanced}. A generator matrix for a $q$-ary code $\mathcal{C}$ is a matrix whose rows constitute a minimal set of generators for $\mathcal{C}$.

The usual inner product of two vectors $\boldsymbol{x}$ and $\boldsymbol{y}$ in $\R^{n}$ is denoted by $\boldsymbol{x} \cdot \boldsymbol{y}$. For each pair of $n$-tuples $\boldsymbol{x} = (x_1, \ldots, x_n)$ and $\boldsymbol{y} = (y_1, \ldots, y_n)$ in $\Z_{q}^{n}$, we define the (Euclidean) semi-inner product between $\boldsymbol{x}$ and $\boldsymbol{y}$ as $\boldsymbol{x} \cdot \boldsymbol{y} := x_1y_1 + \cdots + x_ny_n \in \Z_q$, where $x_iy_i$ denote the usual product over the ring $\Z_q$ for each $i = 1, \ldots, n$. When $q$ is not a prime number, this is not an inner product, since there exist nonzero elements whose product is zero. Given a $q$-ary linear code $\mathcal{C}$, the set $\mathcal{C}^{\perp} := \left\{\boldsymbol{x} \in \Z_q^{n}: \boldsymbol{x} \cdot \boldsymbol{y} = \boldsymbol{0} \hspace{0.1cm} ,\forall \boldsymbol{y} \in \mathcal{C} \right\}$ is always a linear code over $\Z_q$, which is called the dual code of $\mathcal{C}$. If $\mathcal{C}_1$ and $\mathcal{C}_2$ are $q$-ary linear codes such that $\mathcal{C}_2 \subseteq \mathcal{C}_1$, then $\mathcal{C}_1^{\perp} \subseteq \mathcal{C}_2^{\perp}$.

A lattice $\Lambda \subset \R^{n}$ is a discrete additive subgroup of $\R^{n}$. Equivalently, $\Lambda \subset \R^{n}$ is a lattice if, and only if, there exists a set of linearly independent vectors $\boldsymbol{v}_1, \ldots, \boldsymbol{v}_m \in \R^{n}$ such that $\Lambda$ is given by all integer linear combinations of these vectors \cite{cassels2012introduction}.

Under this description, we call the set $\left\{\boldsymbol{v}_1, \ldots, \boldsymbol{v}_{m}\right\}$ a basis of $\Lambda$ and the number $m$, the rank of $\Lambda$. When $m = n$, we say that $\Lambda$ is a full-rank lattice. The matrix $\boldsymbol{M}$ whose columns are the vectors $\boldsymbol{v}_1, \ldots, \boldsymbol{v}_m$ is a generator matrix of $\Lambda$. Two matrices $\boldsymbol{M}_1$ and $\boldsymbol{M}_2$ are generator matrices of the same lattice $\Lambda$ if, and only if, there is a unimodular matrix $\boldsymbol{U}$ (i.e., a matrix with integer entries and $\det \boldsymbol{U} = \pm 1$) such that $\boldsymbol{M}_2 = \boldsymbol{M}_1 \boldsymbol{U}$. Given a generator matrix $\boldsymbol{M}$ for $\Lambda$, we define the associated Gram matrix as $\boldsymbol{\mathcal{G}} = \boldsymbol{M}^{T} \boldsymbol{M}$. The volume of $\Lambda$ is defined as $\vol\Lambda = \sqrt{\det \boldsymbol{\mathcal{G}}}$, where $\boldsymbol{\mathcal{G}}$ is a Gram matrix for $\Lambda$. In this paper, we deal only with full-rank lattices ($m = n$) and, in this case, $ \vol\Lambda = |\det \boldsymbol{M}|$, where $\boldsymbol{M}$ is a generator matrix for $\Lambda$. For a full-rank lattice $\Lambda$, the dual is defined as $\Lambda^{\ast} = \left\{\boldsymbol{y} \in \R^{n} : \boldsymbol{y} \cdot \boldsymbol{x} \in \Z  \hspace{0.1cm}, \forall \boldsymbol{x} \in \Lambda\right\}$. It can be shown that $\boldsymbol{M}$ is a generator matrix for $\Lambda$ if, and only if, $(\boldsymbol{M}^{T})^{-1}$ is a generator matrix for $\Lambda^{\ast}$.

Considering a distance $d$ in $\R^{n}$, we say that two lattices $\Lambda_1, \Lambda_2 \subset \R^{n}$ are $d$-equivalent with respect to a distance $d$ if there exist a number $k \in \R^{\ast}$ and an isometry $\phi$ in $\R^{n}$ with respect to $d$ such that $\Lambda_{2} = k \phi(\Lambda_1)$. Also, the minimum distance of $\Lambda$ with respect to distance $d$ is defined as $d_{d}(\Lambda) := \min \left\{d(\boldsymbol{x}, \boldsymbol{y}): \boldsymbol{x}, \boldsymbol{y} \in \Lambda \text{ and } \boldsymbol{x} \neq \boldsymbol{y}\right\}$. The packing radius $r_{\text{pack}, d}$ of a lattice $\Lambda$, with respect to a distance $d$, is half of the minimum distance of $\Lambda$ relative to this same distance. We consider here the usual $L_{\mathrm{P}}$-distances in $\R^{n}$ and in $\Z_{q}^{n}$ associated with the $L_{\mathrm{P}}$-norm.
The  $L_{\mathrm{P}}$-distance, with $1 \leq \mathrm{P} \leq \infty$, between two elements $\boldsymbol{x}$ and $\boldsymbol{y}$ in $\R^{n}$ is defined as:
\begin{equation*}
    d_{\mathrm{P}}(\boldsymbol{x}, \boldsymbol{y}) :=\left(\displaystyle \sum_{i = 1}^{n} |x_i - y_i|^{\mathrm{P}}\right)^{1/\mathrm{P}} \text{ for } 1 \leq \mathrm{P} < \infty \hspace{0.3cm}\text{ and } \hspace{0.3cm}
    d_{\infty}(\boldsymbol{x}, \boldsymbol{y}) := \max \big\{|x_i - y_i|: \ i = 1, \ldots, n\big\}.
\end{equation*}
Given a lattice $\Lambda \subset \R^{n}$, the minimum $L_{\mathrm{P}}$-distance of $\Lambda$ is defined as
$$d_{\mathrm{P}}(\Lambda) = \min \big\{d_{\mathrm{P}}(\boldsymbol{x}, \boldsymbol{y}): \ \boldsymbol{x}, \boldsymbol{y} \in \Lambda \text{ and } \boldsymbol{x} \neq \boldsymbol{y}\big\}.$$ The Lee distance, introduced by \cite{lee1958some} and \cite{ulrich1957non}, is the induced $L_1$-distance from $\Z$ in $\Z_{q}$ and it is defined as $d_{Lee}(x,y) = \min \left\{\sigma(x - y), q - \sigma(x - y)\right\}$ and for two $n$-tuples $\boldsymbol{x}, \boldsymbol{y} \in \Z_q^{n}$ is given by
\begin{equation*}
    d_{Lee}(\boldsymbol{x}, \boldsymbol{y}) := \displaystyle \sum_{i = 1}^{n} d_{Lee}(x_i, y_i).
\end{equation*}
In addition, the correspondent induced $L_{\mathrm{P}}$-distance from $\Z^{n}$ in $\Z_q^{n}$ (also called \emph{$\mathrm{P}$-Lee distance}) \cite{jorge2013q}, for $\boldsymbol{x}, \boldsymbol{y} \in \Z_q^{n}$ is given by
\begin{equation*}
    d_{\mathrm{P}} (\boldsymbol{x}, \boldsymbol{y}) := \left(\displaystyle \sum_{i = 1}^{n} d_{Lee}(x_i, y_i)^{\mathrm{P}}\right)^{1/\mathrm{P}} \text{ for } 1 \leq \mathrm{P} < \infty  \text{ and }
    d_{\infty}(\boldsymbol{x}, \boldsymbol{y}) := \max \left\{d_{Lee}(x_i,y_i): \ i = 1, \ldots, n\right\}.
\end{equation*}

We denote the $L_{\mathrm{P}}$-norm of a vector $\boldsymbol{x} \in \Z^{n}$ as $||\boldsymbol{x}||_{\mathrm{P}} = d_{\mathrm{P}}(\boldsymbol{x}, \boldsymbol{0})$ and, similarly, the $\mathrm{P}$-Lee norm of $\boldsymbol{y} \in \Z_q^{n}$ as $||\boldsymbol{y}||_{\mathrm{P}} = d_{\mathrm{P}}(\boldsymbol{y}, \boldsymbol{0})$. The minimum $L_{\mathrm{P}}$-distance of a linear code $\mathcal{C} \subseteq \Z_q^{n}$ is defined as $d_{\mathrm{P}}(\mathcal{C}) := \min \left\{d_{\mathrm{P}}(\boldsymbol{x}, \boldsymbol{y}): \boldsymbol{x}, \boldsymbol{y} \in \mathcal{C} \text{ and } \boldsymbol{x} \neq \boldsymbol{y}\right\}$.

For $P = 2$ (Euclidean distance), we use $r_{\text{pack}, d_2}$, $\Delta(\Lambda)$ and $\delta(\Lambda)$ to denote the packing radius, density and center density, respectively. The coding gain and the center density of a full-rank lattice $\Lambda \subset \R^{n}$ are defined, respectively, as
\begin{equation*}
    \gamma(\Lambda) := \dfrac{d_2^{2}(\Lambda)}{(\vol\Lambda)^{2/n}} \hspace{0.3cm} \text{ and } \hspace{0.3cm} \delta(\Lambda) := \dfrac{r_{\text{pack}, d_2}^{n}(\Lambda)}{\vol \Lambda} = 2^{-n}\gamma(\Lambda)^{n/2}.
\end{equation*}

The strong association between lattices in $\Z^{n}$ and linear codes in $\Z_{q}^{n}$ comes from the fact that given a subset $S \subseteq \Z_{q}^{n}$, $\rho^{-1}(S)$ is a lattice if, and only if, $S$ is a $q$-ary linear code \cite{costa2017lattices}. This leads to Construction A  definition \cite{conway2013sphere, zamir2014lattice, costa2017lattices}. Given a linear code $\mathcal{C} \subseteq \Z_q^{n}$, the \emph{Construction A lattice associated with $\mathcal{C}$}, denoted by $\Lambda_A(\mathcal{C})$, is defined as $\Lambda_A(\mathcal{C}) = \rho^{-1}(\mathcal{C}) = \sigma(\mathcal{C}) + q\Z^{n}$. It is shown in \cite{zamir2014lattice} that $\Lambda_A(\mathcal{C})$ is always a full-rank lattice.

\section{Construction D and D': general properties} \label{SectionPropGerais}
In this section, we present some general properties of Construction D and D'. For this, we first need to establish connections between Constructions D and D' and,  subsequently, with Construction A. More details about these connections can be seen in \cite{strey2017lattices, strey2018bounds}.

In the following, the results and definitions cited are adapted versions of \cite{strey2017lattices} using the scaled version of Construction D presented next.

\begin{definition}[Construction D]\label{DefConstD}
 Let $\mathds{Z}^n_q \supseteq \mathcal{C}_1 \supseteq \mathcal{C}_2  \supseteq \cdots  \supseteq \mathcal{C}_a$ be a family of nested linear codes such that $\mathcal{C}_{\ell} = \langle \boldsymbol{b}_1, \ldots,\boldsymbol{b}_{k_\ell}  \rangle$ with $\ell = 1,2,\ldots,a$ for a set of $n$-tuples $\{\boldsymbol{b}_1, \ldots,\boldsymbol{b}_{k_1} \}$ in $\mathds{Z}^n_q$, with integers $k_1 \geq k_2 \geq \cdots \geq k_a \geq 0=:k_{a + 1}$. The lattice $\Lambda_D$ is defined as
$$\Lambda_D = \left\{ q^a\boldsymbol{z}+\sum_{s=1}^{a}\sum_{i=k_{s+1}+1}^{k_s}{\alpha_i^{(s)}q^{a-s}\sigma(\boldsymbol{b}_i)}: \  \boldsymbol{z} \in \mathbb{Z}^n \ \mbox{and} \ 0 \leq \alpha_i^{(s)} < q^s \right\}.$$
Equivalently, we can write
\begin{equation*}
    \Lambda_D = \left\lbrace q^a\boldsymbol{z}+\sum_{s=1}^{a}\sum_{i=k_{s+1}+1}^{k_s}{\alpha_i^{(s)}q^{a-s}\sigma(\boldsymbol{b}_i)}: \  \boldsymbol{z} \in \mathbb{Z}^n \ \mbox{and} \ 0 \leq \alpha_i^{(s)} < \mathcal{O}(\boldsymbol{b}_i)q^{s-1} \right\rbrace,
\end{equation*}
where $\mathcal{O}(\boldsymbol{b}_i)$ is the order of $\boldsymbol{b}_i$ over $\Z_q$ for each $i = 1, \ldots, k_1$.
\end{definition}

\begin{remark}
    The set $\Lambda_D$ is a full-rank lattice in $\R^n$ \cite{strey2017lattices}. Also, when $a=1$, the Construction $D$ coincides with the Construction A. If $q$ is prime, each linear code $\mathcal{C}_{\ell}$ is a vector subspace of $\Z_q^n$ and we can always choose as parameters $k_\ell= \dim{\mathcal{C}_{\ell}}$ $(\ell=1, \ldots, a)$ and $n$-tuples linearly independent $\boldsymbol{b}_1, \ldots,  \boldsymbol{b}_{k_1} \in \Z_q^n$ such that $\mathcal{C}_\ell=\left\langle \boldsymbol{b}_1, \ldots, \boldsymbol{b}_{k_\ell}\right\rangle$. When $q = 2$, the Definition \ref{DefConstD} restricted to these parameters coincides with the original version of the Construction $D$ presented in \cite{barnes1983new,conway2013sphere} without the restriction on the minimum distance.
\end{remark}

\begin{remark}
    It is important to observe that Construction D depends not only on the nested codes as a whole but also on their generators chosen \cite{strey2017lattices}. As an example, consider the chains of nested linear codes
 $\mathcal{C}_2 \subseteq \mathcal{C}_{1} \subseteq \Z_5^{3}$ and $\Hat{\mathcal{C}}_2 \subseteq \Hat{\mathcal{C}}_1 \subseteq \Z_5^{3}$, where $\mathcal{C}_2 = \langle (1,2,0)\rangle$, $\mathcal{C}_1 = \langle (1,2,0),(0,0,1)\rangle$, $\Hat{\mathcal{C}}_2 = \langle (3,1,0)\rangle$ and $\Hat{\mathcal{C}}_1 = \langle (3,1,0), (0,0,1) \rangle$. Note that $\mathcal{C}_1 = \Hat{\mathcal{C}}_1$ and $\mathcal{C}_2 = \Hat{\mathcal{C}}_2$, but the associated Construction D provides different lattices as can be seen from next Theorem \ref{TeoconstDcomoA}, once it guarantees that this construction can be seen as a Construction $A$ where the generator matrices for the associated codes in $\Z_{25}^{3}$ are, respectively,
    \begin{equation*}
    \boldsymbol{G} = \left[ \begin{array}{ccc}
        1 & 2 & 0 \\
        0 & 0 & 5
    \end{array}\right] \hspace{0.5cm} \text{ and} \hspace{0.5cm}
    \boldsymbol{\Hat{G}} = \left[\begin{array}{ccc}
         3 & 1 & 0 \\
         0 & 0 & 5
    \end{array}\right].
\end{equation*}
As described in \cite{costa2017lattices},  generator matrices for the associated Construction $A$ lattices are obtained by the Hermite Normal Form of matrices which have the code generators for the two first columns added by the three columns which vectors $(25,0,0), (0,25,0),(0,0,25)$ are the following:
    \begin{equation*}
    \boldsymbol{M} = \left[ \begin{array}{ccc}
        1 & 0 & 0 \\
        2 & 25 & 0 \\
        0 & 0 & 5
    \end{array}\right] \hspace{0.5cm} \text{ and} \hspace{0.5cm}
    \boldsymbol{\Hat{M}} = \left[\begin{array}{ccc}
         1 & 0 & 0 \\
         17 & 25 & 0 \\
         0 & 0 & 5
    \end{array}\right],
\end{equation*}
That is, $\boldsymbol{M}$ and $\boldsymbol{\Hat{M}}$ are generators matrices of the Constructions $D$ lattices $\Lambda_{D}$ and $\Hat{\Lambda}_D$ obtained from the chains $\mathcal{C}_2 \subseteq \mathcal{C}_{1} \subseteq \Z_5^{3}$ and $\Hat{\mathcal{C}}_2 \subseteq \Hat{\mathcal{C}}_1 \subseteq \Z_5^{3}$ with the above chosen generators, respectively.
Note that $\Lambda_D \neq  \Hat{\Lambda}_D$ since $\boldsymbol{U} = \boldsymbol{\Hat{M}}^{-1}\boldsymbol{M}$ is not unimodular. Moreover, as $d_2(\Lambda_D) = \sqrt{5}$, $d_2(\Hat{\Lambda}_D) = \sqrt{10}$ and $\vol\Lambda_D = 125 = \vol\Hat{\Lambda}_D$, we have that center packing densities of these lattices are $ \delta(\Lambda_D) \approx  0.011$ and  $ \delta(\Hat{\Lambda}_D) \approx  0.032$ , so these lattices are not equivalent. In this case, they have the same volume, but in general, it is not always true. If we consider $\tilde{\mathcal{C}}_2 = \langle (3,1,0),(4,3,0)\rangle$ and $\tilde{\mathcal{C}}_1 = \langle (3,1,0),(4,3,0),(0,0,1) \rangle$ over $\Z_{5}$, the chain remains the same with different choice of generators. Now $d_2(\Lambda_D)= \sqrt{5} = d_2(\tilde{\Lambda}_D)$, but $\vol\tilde{\Lambda}_D = 25$.
\end{remark}

A more natural construction from nested codes, but which does not always produce a lattice, is the Construction $\overline{D}$ also known as Construction by Code Formula \cite{forney1988coset}.

\begin{definition}[Construction $\overline{D}$]  Let $\mathds{Z}^n_q \supseteq \mathcal{C}_1 \supseteq \mathcal{C}_2  \supseteq \cdots  \supseteq \mathcal{C}_a$ be a family of nested linear codes, the set $\Gamma_{\overline{D}}$ is defined as follows
$$\Gamma_{\overline{D}} = q^a\mathds{Z}^n + q^{a-1}\sigma(\mathcal{C}_{1}) +  \cdots + q^{a-i}\sigma(\mathcal{C}_i) + \cdots + q^1\sigma(\mathcal{C}_{a-1}) + \sigma(\mathcal{C}_a).$$
\end{definition}

When $a = 1$, the Construction $\overline{D}$ coincides with the Construction A for linear codes over $\mathds{Z}_q$ and therefore produces a lattice. {\color{black} The set $\Gamma_{\overline{D}} \subseteq \mathds{Z}^n$ is not always a lattice, what leads to define $\Lambda_{\overline{D}}$ as the smallest lattice with respect to the natural inclusion. In this sense, $\Gamma_{\overline{D}} \subseteq \Lambda_{\overline{D}}$ and if $\Lambda \supseteq \Gamma_{\overline{D}}$ is a lattice, then $\Lambda_{\overline{D}} \subseteq \Lambda$. An equivalent description of $\Lambda_{\overline{D}}$ can also be found in {\cite[Thm $8$]{strey2017lattices}}.}

The following theorem states a necessary and sufficient condition for the Construction $\overline{D}$ to be a lattice. For this, in \cite{strey2017lattices} it is proposed an operation in $\mathds{Z}_q^n$ called zero-one addition and denoted by $\ast$, which is defined for each pair of tuples $\boldsymbol{x} = (x_1, \ldots,x_n)$ and $\boldsymbol{y} = (y_1, \ldots,y_n)$ in $\mathds{Z}_q^n$ as
$$ \boldsymbol{x} \ast \boldsymbol{y} = (x_1 \ast y_1, \ldots, x_n \ast y_n),$$
where
$$x_i \ast y_i = \begin{cases}
	0, & \mbox{if } 0 \leq \sigma(x_i) + \sigma(y_i) < q \\
	1, & \mbox{if } q \leq \sigma(x_i) + \sigma(y_i) \leq 2(q-1)	
	\end{cases}$$
for each $i \in \{1, \ldots,n\}$. When $q = 2$ the zero-one addition coincides with the Schur product \cite{kositwattanarerk2014connections}. We say that a family of nested linear codes $\mathds{Z}^n_q \supseteq \mathcal{C}_1 \supseteq \mathcal{C}_2  \supseteq \cdots  \supseteq \mathcal{C}_a$ is closed under the zero-one addition if and only if for any $\boldsymbol{c}_1,\boldsymbol{c}_2 \in \mathcal{C}_\ell$, then $\boldsymbol{c}_1 \ast \boldsymbol{c}_2 \in \mathcal{C}_{\ell-1}$ for all $\ell = 2, \ldots,a$.

\begin{theorem}[{\cite[Thm $4.3$]{strey2017lattices}}]\label{TeoCadeiafechada}
Given a family of nested  linear codes $\mathds{Z}^n_q \supseteq \mathcal{C}_1 \supseteq \mathcal{C}_2  \supseteq \cdots  \supseteq \mathcal{C}_a$, the following statements are equivalent:
\begin{enumerate}
    \item[1.] $\Gamma_{\overline{D}}$ is a lattice.
    \item[2.] $\mathds{Z}^n_q \supseteq \mathcal{C}_1 \supseteq \mathcal{C}_2  \supseteq \cdots  \supseteq \mathcal{C}_a$ is closed under the zero-one addition.
    \item[3.] $\Gamma_{\overline{D}} = \Lambda_{D} = \Lambda_{\overline{D}}$.
\end{enumerate}
\end{theorem}

\begin{remark}
An immediate consequence of the previous theorem is that if $\mathds{Z}^n_q \supseteq \mathcal{C}_1 \supseteq \mathcal{C}_2  \supseteq \cdots  \supseteq \mathcal{C}_a$  is closed under the zero-one addition then Construction D is the same as Construction $\overline{D}$ and, therefore, it depends only on the codes $\mathcal{C}_1, \mathcal{C}_2,  \ldots,  \mathcal{C}_a$ (and not on their generators).
\end{remark}
The next theorem, proposed in {\color{black}\cite{strey2017lattices}}, establishes a relation between Constructions D and A. A version of this result for Construction D from a family of linear codes over $\Z_p$, with $p$ prime, can be also found in {\cite[Prop $2$]{feng2011lattice}}.


\begin{theorem}[{\cite[Thm $3.5$]{strey2017lattices}}]\label{TeoconstDcomoA}
Let $\boldsymbol{G}_1$ be a matrix whose rows are the vectors $\sigma(\boldsymbol{b}_1), \ldots, \sigma(\boldsymbol{b}_{k_1})$ and $\mathcal{C} \subseteq \Z_{q^{a}}^{n}$ the linear $q^a$-ary code generated by the rows of the matrix $\rho_{q^{a}}(\boldsymbol{G})$, where $\boldsymbol{G} =  \boldsymbol{D}\boldsymbol{G}_1$, with $\boldsymbol{D}$ the diagonal matrix given by
	$$d_{jj} = \begin{cases}
	1, & \mbox{for } 1 \leq j \leq k_a; \\
	q, & \mbox{for } k_a < j \leq k_{a-1}; \\
	\vdots  & \\
	q^{a-1}, & \mbox{for } k_2 < j \leq k_1;	
	\end{cases}$$
	Then $\Lambda_D = \Lambda_A(\mathcal{C})$, i.e, $\Lambda_D$ is a $q^a$-ary lattice.
\end{theorem}

The next definition of Construction D' for $q$-ary linear codes \cite{strey2018bounds,strey2017lattices} is an extension of the one presented in \cite{barnes1983new, conway2013sphere}.

\begin{definition}[Construction D']\label{defiD'eleonesio} Let $\mathds{Z}^n_q \supseteq \mathcal{C}_1 \supseteq \mathcal{C}_2  \supseteq \cdots  \supseteq \mathcal{C}_a$ be a family of nested linear codes. Given integers $r_1, r_2, \ldots, r_a$ satisfying $0 \leq r_1 \leq r_2 \leq \cdots \leq r_a$ and a set $\{\boldsymbol{h}_1, \ldots,\boldsymbol{h}_{r_a} \}$ in $\mathds{Z}^n_q$ such that $\mathcal{C}^{\perp}_\ell = \langle \boldsymbol{h}_1, \ldots,\boldsymbol{h}_{r_\ell}  \rangle$ for $\ell = 1,2,\ldots,a$ where $\mathcal{C}^{\perp}_\ell$ is the dual code of $\mathcal{C}_\ell$, the set $\Lambda_{D'}$ consists of all vectors $\boldsymbol{x} \in \mathds{Z}^n$ such that
$$ \boldsymbol{x} \cdot \sigma(\boldsymbol{h}_j) \equiv 0 \mod  q^{i+1}$$
for each pair of integers $(i,j)$ satisfying $0 \leq i < a$ and $r_{a-i-1} < j \leq r_{a-i}$, where $r_0 := 0$.
\end{definition}

\begin{remark}\label{lemamatrizanivelq}
The congruence equations in Definition \ref{defiD'eleonesio} can be rewritten via a check matrix denoted by $\boldsymbol{H}$, providing an equivalent characterization to Construction D' which is presented in \cite{sadeghi2006low, strey2017lattices, zhou2022construction}. Indeed, let $\boldsymbol{H}_a$ be the matrix whose rows are the vectors $\sigma(\boldsymbol{h}_1), \ldots, \sigma(\boldsymbol{h}_{r_a})$ and  $\mathcal{C}$ the $q^a$-ary linear code  generated by the rows of $\rho_{q^{a}}(\boldsymbol{H})$, where $ \boldsymbol{H} = \boldsymbol{D}\boldsymbol{H}_a$, with $\boldsymbol{D}$ the diagonal matrix defined as in Theorem \ref{TeoconstDcomoA}, with $k_{i} = r_{a - i + 1}$ for each $i = 1, \ldots, a$.
Then, $ \Lambda_{D'} = \big\{ \boldsymbol{x} \in \mathbb{Z}^n: \ \boldsymbol{H} \boldsymbol{x} \equiv \boldsymbol{0} \mod q^a \big\}$. Throughout this text, we call a matrix $\boldsymbol{H}$ a check matrix (also known as an $a$-level matrix) associated to $\Lambda_{D'}$ as in \cite{sadeghi2006low, silva2020multilevel}. We point out that there exists another definition of a check matrix associated with low-density lattice codes, which is presented in \cite{sommer2009shaping} and used in \cite{zhou2022construction}.
\end{remark}

In general, Construction D' depends on the choice of code generators, as can be seen in the example below.

\begin{example}\label{Exemplodo42e24}
Consider $\mathbb{Z}^2_6 \supseteq \mathcal{C}_1 = \Hat{\mathcal{C}}_1 \supseteq \mathcal{C}_2 = \Hat{\mathcal{C}}_2$ a family of nested linear codes such that $\mathcal{C}^{\perp}_1 = \langle(4,2) \rangle \subseteq \mathcal{C}^{\perp}_2 = \langle(4,2),  (0,1) \rangle$ and $\Hat{\mathcal{C}}^{\perp}_1 = \langle(2,4)\rangle \subseteq \Hat{\mathcal{C}}^{\perp}_2 = \langle(2,4), (0,1)\rangle.$
As in Remark \ref{lemamatrizanivelq}, we have that
    \begin{equation*}
    \boldsymbol{H} = \left[ \begin{array}{cc}
        4 & 2 \\
        0 & 6
    \end{array}\right] \hspace{0.5cm} \text{and} \hspace{0.5cm}
    \boldsymbol{\Hat{H}} = \left[\begin{array}{cc}
         2 & 4 \\
         0 & 6
    \end{array}\right]
\end{equation*}
are check matrices of $\Lambda_{D'}$ and ${\color{black}\hat{\Lambda}_{D'}}$, respectively.
Equivalently, we have
\begin{eqnarray*}
    \Lambda_{D'} &=&  \left\{(x, y) \in \mathbb{Z}^2: \ 4x+2y \equiv 0 \!\!\mod36 \ \ \mbox{and} \ \ 6y \equiv 0 \!\! \mod 36 \right\} \    \mbox{and} \\
  \Hat{\Lambda}_{D'} &=&  \left\{(x, y) \in \mathbb{Z}^2: \ 2x+4y  \equiv 0 \!\!\mod36 \ \ \mbox{and} \ \ 6y \equiv 0\!\! \mod 36 \right\}.
\end{eqnarray*}
Solving each system of equations, we get generator matrices for $\Lambda_{D'}$ and $\Hat{\Lambda}_{D'}$, respectively, given by
    \begin{equation*}
     \boldsymbol{M} = \left[ \begin{array}{cc}
       9 & -3  \\
        0 & 6
    \end{array}\right] \hspace{0.5cm} \text{ and} \hspace{0.5cm}
    \Hat{\boldsymbol{M}} = \left[\begin{array}{cc}
         18 & -12 \\
         0 & 6
    \end{array}\right].
\end{equation*}
Since $\vol\Lambda_{D'} =  |\det \boldsymbol{M}| = 54 \neq 108 =|\det \hat{\boldsymbol{M}}| = \text{vol }\Hat{\Lambda}_{D'}$, these lattices are different (and also non equivalents) and then in this case the Construction D' depends on the choice of the generators.
\end{example}

The Construction D' is connected with the Construction D using codes which are dual of the original ones \cite{strey2017lattices}.

\begin{definition}\label{DefDdual}
Let $\mathds{Z}^n_q \supseteq \mathcal{C}_1 \supseteq \mathcal{C}_2  \supseteq \cdots  \supseteq \mathcal{C}_a$ be a family of nested linear codes, parameters $r_1, r_2, \ldots, r_a$ satisfying $0 \leq r_1 \leq r_2 \leq \cdots \leq r_a$ and $n$-tuples $\boldsymbol{h}_1, \ldots,\boldsymbol{h}_{r_a}$ in $\mathds{Z}^n_q$ such that $\mathcal{C}^{\perp}_\ell = \langle \boldsymbol{h}_1, \ldots,\boldsymbol{h}_{r_\ell}  \rangle$ for $\ell = 1,2,\ldots,a$. We define $\Lambda_{D^{\perp}}$ as the lattice obtained  via Construction $D$ from a family of nested linear codes $ \mathcal{C}_1^{\perp} \subseteq \mathcal{C}_2^{\perp} \subseteq \cdots  \subseteq \mathcal{C}_a^{\perp} \subseteq \mathds{Z}^n_q$.
\end{definition}

\begin{theorem}[{\cite[Thm $1$]{strey2018bounds}} ]\label{T42} Let $\mathds{Z}^n_q \supseteq \mathcal{C}_1 \supseteq \mathcal{C}_2  \supseteq \cdots  \supseteq \mathcal{C}_a$ be a family of nested linear codes. Then, $\Lambda_{D'} = q^a\Lambda_{D^{\perp}}^{*}$.
\end{theorem}

\begin{remark}
     If we consider $\Lambda_{D^{\perp}}$ the lattice obtained  via Construction $D$ as in Definition \ref{DefDdual}, we have an analogous result of Theorem \ref{TeoCadeiafechada} {\cite[Cor $1$]{strey2018bounds}}. In other words, if the chain of dual codes is closed under the zero-one addition, then Construction D' does not depend on the choice of the code generators. One can observe that lattices obtained in Example \ref{Exemplodo42e24} are distinct and the chain of dual codes is not closed under the zero-one addition since $(4,2) \ast (4,2) = (1,0) \notin \mathcal{C}_2^{\perp}$.
\end{remark}

For the next result, we connect Construction D' with Construction A and, from this connection, we describe how to obtain a generator matrix for lattice $\Lambda_{D'}$ in the general case, calculate its volume, and show an expression for its minimum distance. From Theorem \ref{T42}, we have that $\boldsymbol{M}$ is a generator matrix of $\Lambda_{D^{\perp}}$ if and only if $q^{a}(\boldsymbol{M}^{-1})^{T}$ is a generator matrix for $\Lambda_{D'}$. By applying Theorem \ref{TeoconstDcomoA}, we get the following results.

\begin{corollary}\label{teomatrizgeradoraD'}
Let $\Lambda_{D'}$ be the lattice obtained via Construction $D'$ from Definition \ref{defiD'eleonesio}. Then, $\Lambda_{D'} = q^a\Lambda_{A}^{*}(\mathcal{C})$, where $\mathcal{C} \subseteq \Z_{q^a}^n$ is the linear code generated by the rows of matrix $\rho_{q^{a}}(\boldsymbol{H})$, with $\boldsymbol{H}$ as in Remark \ref{lemamatrizanivelq}.
\end{corollary}

\begin{corollary}[Generator matrix for $\boldsymbol{\Lambda}_{D'}$]\label{cormatrizgeradoraD'}
A generator matrix for $\Lambda_{D'}$ is given by
\begin{equation*}
    \boldsymbol{M} = q^{a} \left(\boldsymbol{B}^{-1}\right)^{T},
\end{equation*}
where $\boldsymbol{B}$ is the Hermite Normal Form (HNF) of
$\left[\begin{array}{cccc}
         \boldsymbol{H}^{T} & q^{a}\boldsymbol{e_1} & \ldots & q^a\boldsymbol{e_n}  \\
    \end{array}\right]$ {\color{black} and $\{\boldsymbol{e}_1, \ldots, \boldsymbol{e}_{n}\}$ is the} canonical basis of $\R^{n}$.
\end{corollary}

\begin{proof}
From Corollary \ref{teomatrizgeradoraD'}, we have that $\Lambda_{D'} = q^a \Lambda_{A}^{\ast}(\mathcal{C})$, where $\mathcal{C} \subseteq \Z_{q^a}^n$ is the linear code obtained for the rows of the matrix $\rho_{q^{a}}(\boldsymbol{H})$. Since $\mathcal{C}$ is a linear code in $\Z_{q^a}^n$, from {\cite[Prop. $3.3$]{costa2017lattices}}, it follows that $\Lambda_{A}(\mathcal{C})$ has a generator matrix (in column form) given by the Hermite Normal Form of
\begin{eqnarray*}
    \left[\begin{array}{cccc}
         \boldsymbol{H}^{T} & q^a \boldsymbol{e}_1 & \ldots & q^a\boldsymbol{e}_n \\
    \end{array}\right].
\end{eqnarray*}
Denote $\boldsymbol{B}$ the Hermite Normal Form from the previous matrix. Then, a generator matrix for $\Lambda_{A}^{\ast}(\mathcal{C})$ is $\left(\boldsymbol{B}^{-1}\right)^{T}$ \cite{costa2017lattices}. Finally, since $\Lambda_{D'} = q^{a} \Lambda_A^{\ast}(\mathcal{C})$, we conclude that $q^a \left(\boldsymbol{B}^{-1}\right)^{T}$ is a generator matrix for $\Lambda_{D'}$.
\end{proof}

\begin{corollary}[Volume of $\boldsymbol{\Lambda}_{D'}$]\label{CorVolumeDlinhaviaA}
The volume of $\Lambda_{D'}$ is given by
\begin{eqnarray*}
    {\color{black}\vol \ } \Lambda_{D'} &=& |\det \boldsymbol{M}| = |\mathcal{C}|,
\end{eqnarray*}
where $\mathcal{C}$ is the $q^{a}$-ary linear code  of Corollary \ref{teomatrizgeradoraD'} and $|\mathcal{C}|$ is the cardinality  of $\mathcal{C}$. In particular, an upper bound for the volume of $\Lambda_{D'}$ is ${\color{black}\vol \ } \Lambda_{D'} \leq q^{ar_a}$. Furthermore, if the rows of $\rho_{q^{a}}(\boldsymbol{H})$ are linearly independent in $\Z_{q^{a}}^{n}$, then
we have equality. \end{corollary}
\begin{proof}
For the description of $\Lambda_{D'}$ in Corollary \ref{teomatrizgeradoraD'} and as $\Lambda_{D'}$ is a full-rank lattice in $\R^n$, we have
\begin{eqnarray*}
    \text{vol } \Lambda_{D'} &=& q^{an}\text{ vol } \Lambda_{A}^{\ast}(\mathcal{C}) = \dfrac{q^{an}}{\text{vol } \Lambda_{A}(\mathcal{C})} = \dfrac{q^{an}}{{q^{an}}/{|\mathcal{C}|}} = |\mathcal{C}|.
\end{eqnarray*}
{\color{black} To finish the proof,  it is enough to observe that $\rho_{q^{a}}(\boldsymbol{H})$ is a matrix $r_a\times n$ and  $\mathcal{C}$ is the $q^a$-ary linear code generated by the rows of this matrix (Remark \ref{lemamatrizanivelq}), which can be linearly independent or not.}
\end{proof}

Given a {\color{black} linear code} $\mathcal{C} \subseteq \Z_q^{n}$, we can see that $\Lambda_{A}(\mathcal{C}^{\perp}) = q \Lambda_{A}^{\ast}(\mathcal{C})$ \cite{micciancio2009lattice, jorge2012reticulados}. In fact, we have
{\color{black}
\begin{eqnarray*}
    \boldsymbol{x} \in \rho^{-1}(\mathcal{C}^{\perp}) & \ \Leftrightarrow \ & \rho(\boldsymbol{x}) \cdot \boldsymbol{y} = 0, \ \forall \boldsymbol{y} \in \mathcal{C} \ \Leftrightarrow  \ \rho(\boldsymbol{x}) \cdot \rho(\boldsymbol{h})  = 0, \ \forall \boldsymbol{h} \in \rho^{-1}(\mathcal{C})\\
    & \ \Leftrightarrow \ &  \boldsymbol{x} \cdot \boldsymbol{h}  = q k, \ \mbox{for some} \ k \in \Z \ \Leftrightarrow \ \boldsymbol{x} \in q \Lambda_A^{\ast}(\mathcal{C}).
\end{eqnarray*}}

{\color{black} The next theorem is straightforward} from Corollary \ref{teomatrizgeradoraD'} replacing $q^{a}\Lambda_A^{\ast}(\mathcal{C})$ by $\Lambda_{A}(\mathcal{C}^{\perp})$, and describes $\Lambda_{D'}$ as a Construction A.

\begin{theorem}
\label{propD'viaconstrucaoA}
Under the notation of Definition \ref{defiD'eleonesio}, we can express the lattice $\Lambda_{D'}$ as
\begin{equation*}
    \Lambda_{D'} = \sigma_{q^{a}}(\mathcal{C}^{\perp}) + q^a \Z^n = \Lambda_A(\mathcal{C}^{\perp}),
\end{equation*}
where $\mathcal{C}^{\perp} = \Lambda_{D'} \cap [0,q^a)^n$ is the dual code $q^a$-ary with check matrix $\rho_{q^{a}}(\boldsymbol{H})$,
where $\boldsymbol{H}$ is as in Remark \ref{lemamatrizanivelq}.
\end{theorem}

{\color{black}\begin{remark}\label{volviaConstA}
It follows directly from the above theorem and from {\cite[Prop. $3.2$]{costa2017lattices}} that $\text{vol } \Lambda_{D'} = q^{an}/|\mathcal{C}^{\perp}|$. It should be noticed that this result is also presented in {\cite[Eq. $9$]{silva2020multilevel}} for Construction D' from a chain of binary codes.
\end{remark}}

We can calculate the $L_{\mathrm{P}}$-distances and volume of the lattice $\Lambda_{D'}$ by using its association with Construction A and results of \cite{costa2017lattices}.

\begin{corollary}[Minimum distance of $\Lambda_{D'}$]\label{cormindistance}
Consider the distance $L_{\mathrm{P}}$, with $1 \leq \mathrm{P} \leq \infty$. Then, the minimum distance of $\Lambda_{D'}$ is
\begin{equation*}
    d_{\mathrm{P}} (\Lambda_{D'}) = \min \left\{d_{\mathrm{P}}(\mathcal{C}^{\perp}), q^a\right\}.
\end{equation*}
\end{corollary}

Motivated by the work of \cite{zhou2022construction}, we investigate an alternative way to obtain a generator matrix for the lattice $\Lambda_{D'}$. We finish this section with a generator matrix for this lattice without using the Hermite Normal Form, under specific conditions. This result is a direct consequence of Theorem \ref{propD'viaconstrucaoA} and extends, for $q$-ary linear codes and a larger number of lattices, the Proposition $1$, proposed and demonstrated in \cite{zhou2022construction} to binary linear codes, with appropriate notation adjustments.

\begin{corollary}[Generator matrix for $\boldsymbol{\Lambda}_{D'}$]\label{CorGeneralizaZhou}
Let $\mathcal{C}_a \subseteq \mathcal{C}_{a - 1} \subseteq \cdots \subseteq \mathcal{C}_{1} \subseteq \mathbb{Z}_{q}^{n}$ be a family of nested linear codes. Given $r_1, \ldots, r_{a} \in \mathbb{N}$ satisfying  $r_0 : = 0 \leq r_1 \leq r_2 \leq \cdots \leq r_{a} = n$ and  $\left\{\boldsymbol{h}_1, \ldots, \boldsymbol{h}_{n}\right\} \subseteq\mathbb{Z}_q^n$ such that $\mathcal{C}^{\perp}_{\ell} = \langle  \boldsymbol{h}_{1}, \ldots, \boldsymbol{h}_{r_{\ell}}\rangle$ for $1 \leq \ell \leq a$, where $\mathcal{C}_\ell^{\perp}$ is the dual of $\mathcal{C}_\ell$. Consider the lattice obtained via Construction $D'$ from that chain using the above parameters. Let $\boldsymbol{H} = \boldsymbol{D}\boldsymbol{H}_a$ as  in the Corollary \ref{teomatrizgeradoraD'}. Suppose that $\sigma\left(\boldsymbol{h}_1\right), \ldots, \sigma\left(\boldsymbol{h}_{n}\right)$ are linearly independent.
Then, $q^a\boldsymbol{H}^{-1}$ is a matrix of integer entries if and only if  $q^{a} {\boldsymbol{H}^{-1}}$ is a generator matrix of the lattice $\Lambda_{D'}$.
\end{corollary}
\begin{proof}
$(\Rightarrow)$ Since $r_a = n$ and $\sigma\left(\boldsymbol{h}_{1}\right), \ldots, \sigma\left(\boldsymbol{h}_{n}\right)$ are linearly independent,  we have that $\boldsymbol{D}$ and $\boldsymbol{H}_{a}$ are invertible matrices of order $n$. Then, $\boldsymbol{H} = \boldsymbol{D} \boldsymbol{H}_a$ is also an invertible matrix of order $n$.
So,
\begin{eqnarray*}
    \boldsymbol{x} \in \Lambda_{D'} &\Leftrightarrow& \boldsymbol{x} \in \Z^n \hspace{0.2cm} \text{ and } \hspace{0.2cm} \boldsymbol{H}\boldsymbol{x} \equiv \boldsymbol{0} \mod q^{a} \\
    &\Leftrightarrow& \boldsymbol{x} \in \Z^n \hspace{0.2cm} \text{ and } \hspace{0.2cm} \boldsymbol{x} = q^{a} \boldsymbol{H}^{-1}\boldsymbol{z} \text{ for some } \boldsymbol{z}\in \Z^n\\
     &\Leftrightarrow& \boldsymbol{x} \in \left\{\left(q^{a}{\boldsymbol{H}}^{-1}\right) \boldsymbol{z}: \ \boldsymbol{z} \in \Z^n\right\},
\end{eqnarray*}
since all entries of the matrix $q^a\boldsymbol{H}^{-1}$ are integers. Therefore, $q^{a} {\boldsymbol{H}^{-1}}$ is a generator matrix for $\Lambda_{D'}$.\\
$(\Leftarrow)$  The reciprocal is immediate, since $\Lambda_{D'}$ is an integer lattice.
\end{proof}

\begin{remark}{\color{black}
    The condition} that the matrix  $q^{a} {\boldsymbol{H}^{-1}}$ has integer entries is not always satisfied. Consider $\Z^2_3 \supseteq \mathcal{C}_1 \supseteq \mathcal{C}_2$ such that $\mathcal{C}_1^{\perp}  = \left\langle (1,0) \right\rangle \ \ \mbox{and} \ \
\mathcal{C}_2^{\perp}  = \left\langle (1,0), (0,2) \right\rangle.$ {\color{black} In this example, we have}
    \begin{equation*}
    \boldsymbol{D} = \left[ \begin{array}{cc}
        1 & 0 \\
        0 & 3
    \end{array}\right] \hspace{0.5cm} \text{ and} \hspace{0.5cm}
    \boldsymbol{H}_2 = \left[\begin{array}{cc}
         1 & 0 \\
         0 & 2
    \end{array}\right],
\end{equation*}
but
$$3^2\left(\boldsymbol{D}\boldsymbol{H}_2\right)^{-1}=  9 \left[\begin{array}{cc}
         1 & 0 \\
         0 & 6
    \end{array}\right]^{-1}
    = \left[\begin{array}{cc}
         9 & 0\\
         0 & 3/2
    \end{array}\right]$$
does not have integer entries and, therefore, cannot generate the lattice $\Lambda_{D'}$.
\end{remark}

\begin{example}
Let $\mathbb{Z}^2_6 \supseteq \mathcal{C}_1 \supseteq \mathcal{C}_2$ be a family of nested linear codes such that
$\mathcal{C}_1  =  \left\langle (1, 2) \right\rangle$ and
$\mathcal{C}_2  =  \left\langle (2, 4) \right\rangle.$
Note that $\mathcal{C}_1^{\perp} = \langle (4,1) \rangle$ and $\mathcal{C}_2^{\perp} = \langle (4,1), (3,0)\rangle$. Applying Construction D', we get
$$\Lambda_{D^{'}} =  \bigg\{(x, y) \in \mathbb{Z}^2: \ 4x+y \equiv 0 \!\!\mod36 \ \ \mbox{and} \ \ 3x \equiv 0 \!\! \mod 6 \bigg\}.$$
Since
    \begin{equation*}
    \boldsymbol{D} = \left[ \begin{array}{cc}
         1 & 0 \\
        0 & 6
    \end{array}\right] \hspace{0.5cm} \text{ and } \hspace{0.5cm}
    \boldsymbol{H}_2 = \left[\begin{array}{cc}
         4 & 1 \\
         3 & 0
    \end{array}\right],
\end{equation*}
$$\boldsymbol{B} := 6^2\left(\boldsymbol{D}\boldsymbol{H}_2\right)^{-1}= 36  \left[\begin{array}{cc}
         4 & 1 \\
         18 & 0
    \end{array}\right]^{-1} = 36 \left[\begin{array}{cc}
         0 & 1/18\\
         1 & -2/9
    \end{array}\right] = \left[\begin{array}{cc}
         0 & 2\\
         36 & -8
    \end{array}\right],$$
has integer entries, Corollary \ref{CorGeneralizaZhou} guarantees that $\boldsymbol{B}$ is a generator matrix of $\Lambda_{D^{'}}$.
\end{example}

\begin{remark}\label{CorCasoParticularZhou}
Consider the notation of Corollary \ref{CorGeneralizaZhou}, with $r_a = n$ and $\boldsymbol{h}_1, \ldots, \boldsymbol{h}_n$ linearly independent over $\Z_q$. Suppose that $\rho_{q^{a}}(\boldsymbol{H}_a)$ is a unimodular matrix over $\Z_q$. Then $q^{a} \boldsymbol{H}^{-1}$ is a generator matrix for the lattice $\Lambda_{D'}$ and
\begin{equation*}
    \text{vol } \Lambda_{D'} = |\det \boldsymbol{G}| = q^{an} |\det \boldsymbol{H}^{-1}| = q^{an} |\det \boldsymbol{H}_{a}^{-1}| \cdot |\det \boldsymbol{D}^{-1}| = q^{an} |\det \boldsymbol{D}^{-1}| = \displaystyle \prod_{i = 0}^{a - 1} \left(q^{an -i}\right)^{r_{i + 1} - r_{i}}.
\end{equation*}
We highlight that when $q=2$ this result corresponds exactly to  Proposition  $1$ of \cite{zhou2022construction}. In \cite{zhou2022construction}, the set of vectors $\{\boldsymbol{h}_1, \ldots, \boldsymbol{h}_{r_a}\}$ is completed in such a way that $\boldsymbol{H}_a$ is a unimodular matrix.
\end{remark}

\section{Volume and Minimum Distance of Construction \texorpdfstring{$D$}{D} and \texorpdfstring{$D'$}{D'}}

In Section \ref{SectionPropGerais}, we relate Constructions D and D' with Construction A in such a way that the volume and distance of these lattices can be described through this association. Other descriptions of these lattice parameters under special conditions will be presented in this section.

\subsection{Volume}

The next theorem provides an upper bound for the cardinality of a linear code over $\Z_q$ whose generator matrix can be written as in Theorem \ref{TeoconstDcomoA}.

\begin{theorem}\label{num_max_pontos}
Let $\Lambda_D = \Lambda_A(\mathcal{C})$ be the lattice obtained from Construction D as in Definition \ref{DefConstD}, where $\mathcal{C}$ is the $q^{a}$-ary code generated by the rows of the matrix $\rho_{q^{a}}(\boldsymbol{G})$ as in Theorem \ref{TeoconstDcomoA}. Then, the cardinality of $\mathcal{C}$ satisfies
$$ |\mathcal{C}| = \left|\Lambda_D \cap [0,q^a)^n\right|  \leq \prod\limits_{s=1}^{a} \left(\prod\limits_{i=k_{s+1}+1}^{k_s}{\mathcal{O}(\boldsymbol{b}_i)q^{s-1}}\right) =  \dfrac{q^{\sum\limits_{\ell=1}^{a} k_{\ell}} }{\prod\limits_{i=1}^{k_1}{\dfrac{q}{\mathcal{O}(\boldsymbol{b}_i)}}}= q^{\sum\limits_{\ell=2}^{a} k_{\ell}} \prod\limits_{i=1}^{k_1}{\mathcal{O}(\boldsymbol{b}_i)},$$
and, hence, the volume of Construction $D$ satisfies
\begin{equation*}%
     \vol\Lambda_D {\ \color{black} \geq  \ } q^{an - \sum\limits_{\ell=1}^{a} k_{\ell}} \left(\prod\limits_{i=1}^{k_1}{\dfrac{q}{\mathcal{O}(\boldsymbol{b}_i)}} \right).
\end{equation*}
 {\color{black}  Furthermore, if $\boldsymbol{b}_1, \ldots,\boldsymbol{b}_{k_1}$ are linearly independent over $\Z_{q}$, then
$$|\mathcal{C}| = \left| \Lambda_D \cap [0,q^a)^n\right| = q^{\sum\limits_{\ell=1}^{a} k_{\ell}}  \hspace{0.15cm}  \mbox{and} \hspace{0.15cm} \vol\Lambda_D = q^{an - \sum\limits_{\ell=1}^{a} k_{\ell}}.$$}
\end{theorem}

\begin{proof}
It is enough to prove the first upper bound since the second is a direct consequence of {\color{black} {\cite[Prop $3.2$]{costa2017lattices}} and Theorem \ref{TeoconstDcomoA}}. {\color{black} By Definition \ref{DefConstD},
$$\Lambda_D \cap [0,q^{a})^{n} = \left\lbrace \sum_{s=1}^{a}\sum_{i=k_{s+1}+1}^{k_s}{\alpha_i^{(s)}q^{a-s}\sigma(\boldsymbol{b}_i)} \mod q^a: \  0 \leq \alpha_i^{(s)} < \mathcal{O}(\boldsymbol{b}_i)q^{s-1} \right\rbrace.$$
In other words, the vectors of $\Lambda_D$ inside the box $[0,q^{a})^{n}$} are completely determined by the choices of $\alpha_i^{(s)}$, where $k_{s + 1} + 1 \leq i \leq k_{s}$ and $1 \leq s \leq a$. Also, each $\alpha_i^{(s)}$ can be chosen in $\mathcal{O}(\boldsymbol{b}_i) q^{s - 1}$ different ways. Since the choices of $\alpha_i^{(s)}$ are independent, the Fundamental Counting Principle states that
$$|\mathcal{C}| = \left| \Lambda_D \cap [0,q^a)^n\right|  \leq \prod\limits_{s=1}^{a} \left(\prod\limits_{i=k_{s+1}+1}^{k_s}\mathcal{O}(\boldsymbol{b}_i)q^{s-1}\right).$$
On the other hand, some calculations provide
\begin{eqnarray*}
    \prod\limits_{s=1}^{a} \left(\prod\limits_{i=k_{s+1}+1}^{k_s}{\mathcal{O}(\boldsymbol{b}_i)q^{s-1}}\right)
    & = &
    \prod\limits_{s=1}^{a}{\left(q^{s-1}\right)^{k_s-k_{s+1}}} \cdot \prod\limits_{i=1}^{k_1}{\mathcal{O}(\boldsymbol{b}_i)} = \dfrac{\prod\limits_{s=1}^{a}{\left(q^{s}\right)^{k_s-k_{s+1}}} \cdot \prod\limits_{i=1}^{k_1}{\mathcal{O}(\boldsymbol{b}_i)}}{\prod\limits_{s=1}^{a}{q^{k_s-k_{s+1}}}}\\
    & = & \dfrac{\prod\limits_{s=1}^{a}{\left(q^{s}\right)^{k_s-k_{s+1}}} \cdot \prod\limits_{i=1}^{k_1}{\mathcal{O}(\boldsymbol{b}_i)}}{q^{k_1}} = \dfrac{\prod\limits_{s=1}^{a}{\left(q^{s}\right)^{k_s-k_{s+1}}} }{\prod\limits_{i=1}^{k_1}{\dfrac{q}{\mathcal{O}(\boldsymbol{b}_i)}}} = \dfrac{q^{\sum\limits_{\ell=1}^{a} k_{\ell}} }{\prod\limits_{i=1}^{k_1}{\dfrac{q}{\mathcal{O}(\boldsymbol{b}_i)}}}=q^{\sum\limits_{\ell=2}^{a} k_{\ell}} \prod\limits_{i=1}^{k_1}{\mathcal{O}(\boldsymbol{b}_i)},
\end{eqnarray*}
	by using that
	  \begin{eqnarray*}
	\prod\limits_{s=1}^{a}{\left(q^{s}\right)^{k_s-k_{s+1}}} & = &  q^{k_1-k_2}\left(q^2\right)^{k_2-k_3}\left(q^3\right)^{k_3-k_4}\left(q^4\right)^{k_4-k_5}\cdots \left(q^a\right)^{k_a-k_{a+1}}\\	
      & = & q^{k_1+k_2+k_3+k_4+ \cdots
      +k_a} =   q^{\sum\limits_{\ell=1}^{a} k_{\ell}}.
     \end{eqnarray*}
     For the case where $\boldsymbol{b}_1, \ldots,\boldsymbol{b}_{k_1}$ are linearly independent over $\Z_{q}$, we have $\mathcal{O}(\boldsymbol{b}_i) = q$ for every $i = 1, \ldots, k_1$. From this hypothesis, we also have that different choices of $\alpha_{i}^{(s)}$ yields $n$-tuples that are necessarily distinct inside the box $\Lambda_D \cap [0,q^{a})^{n}$. This proves the expression obtained for $|\mathcal{C}|$. For the volume of $\Lambda_D$, it is enough to apply {\color{black} Theorem \ref{TeoconstDcomoA} and {\cite[Prop $3.2$]{costa2017lattices}}}
\end{proof}

The particular case of Theorem \ref{num_max_pontos}, where the generators are linearly independent over $\Z_{q}$, is stated in {\cite[Thm $3.4$]{strey2017lattices}} for $q$-ary codes and in \cite{barnes1983new, conway2013sphere}, for binary codes (Theorem $1$ and Theorem $13$, respectively).

\begin{remark}
Under the notation used, if $\mathcal{C}_1 = \langle \boldsymbol{b}_1, \ldots, \boldsymbol{b}_{k_1}\rangle$ and $\Hat{\mathcal{C}}_1 =  \langle \Hat{\boldsymbol{b}}_1, \ldots,\Hat{\boldsymbol{b}}_{k_1}\rangle$ are both generated by $k_1$ linearly independent $n$-tuples over $\Z_q$, then both associated Construction D lattices will have the same volume. The following example illustrates this fact in a case where the lattices are not equivalent.
\end{remark}

\begin{example}\label{excomordemdiferentedeplasConstD}
Consider $\mathcal{C}_2 \subseteq \mathcal{C}_1 \subseteq \Z_{6}^{2}$ and $\Hat{\mathcal{C}}_2 \subseteq \Hat{\mathcal{C}}_1 \subseteq \Z_6^{2}$, where $\mathcal{C}_2 = \langle (1,5)\rangle$, $\mathcal{C}_1 = \langle (1,5), (4,1)\rangle = \Hat{\mathcal{C}}_1$ and $\Hat{\mathcal{C}}_2 = \langle(4,1)\rangle$.
Let us denote $\Lambda_{D}$ and $\Hat{\Lambda}_{D}$, respectively, as the lattices obtained from these chains. From Theorem \ref{TeoconstDcomoA}, we have $\Lambda_D = \Lambda_A(\mathcal{C})$ and $\Hat{\Lambda}_{D} = \Lambda_A(\Hat{\mathcal{C}})$, where $\mathcal{C}$ and $\Hat{\mathcal{C}}$ are the $6^{2}$-ary {\color{black}linear} codes generated, respectively, by the rows of the matrices
\begin{equation*}
    \boldsymbol{G} = \left[\begin{array}{cc}
        1 & 5 \\
        24 & 6
    \end{array}\right] \hspace{0.2cm} \text{ and } \hspace{0.2cm} \Hat{\boldsymbol{G}} = \left[\begin{array}{cc}
        4 & 1 \\
        6 & 30
    \end{array}\right].
\end{equation*}
Thus, $\Lambda_{D}$ and $\Hat{\Lambda}_{D}$ are generated, respectively, by
\begin{equation*}
    \boldsymbol{M} = \left[\begin{array}{cc}
        1 & 0 \\
        5 & 6
    \end{array}\right] \hspace{0.2cm} \text{ and } \hspace{0.2cm} \Hat{\boldsymbol{M}} = \left[\begin{array}{cc}
        2 & 0 \\
        2 & 3
    \end{array}\right].
\end{equation*}
{\color{black} Therefore, $\vol \Lambda_{D} = \vol \Hat{\Lambda}_{D} = 6$}, as expected by Theorem \ref{num_max_pontos}. However, it is easy to see that these lattices are non equivalent since they have the same volume but different minimum (Euclidean) distances.
\end{example}

In the upcoming discussion, we will present a sufficient condition to achieve equality in Theorem \ref{num_max_pontos} even for tuples that are not linearly independent. This extends Corollary $3.8$ of \cite{strey2017lattices}.

\begin{theorem}\label{T2}
Let $\boldsymbol{b}_1, \ldots, \boldsymbol{b}_{k_1} \in \Z_{q}^{n}$ be nonzero $n$-tuples such that:
\begin{itemize}
    \item[1.] $\mathcal{C}_{\ell} =\langle \boldsymbol{b}_1,\ldots, \boldsymbol{b}_{k_\ell} \rangle $, for each $\ell = 1, 2, \ldots, a$.
    \item[2.] Some row permutation of the matrix $\boldsymbol{M}$ whose rows are $\boldsymbol{b}_1, \ldots, \boldsymbol{b}_{k_1}$ forms an ``upper triangular'' (respectively, ``lower triangular'') matrix in the row echelon form.
    \item[3.] The first nonzero component (respectively, the last component) of each vector $\sigma(\boldsymbol{b}_i)$, with $i = 1, \ldots, k_1$, divides $q$ as well as all the other components of this vector.
\end{itemize}
 Let $\Lambda_{D}$ be the lattice obtained from the chain $\mathcal{C}_a \subseteq \cdots \subseteq \mathcal{C}_1 \subseteq \Z_q^{n}$ via Construction $D$ under this choice of parameters. Then, if $k_1 = n$, it holds
 $$|\mathcal{C}| = \left| \Lambda_D \cap [0,q^a)^n\right| =  \dfrac{q^{\sum\limits_{\ell=1}^{a} k_{\ell}} }{\prod\limits_{i=1}^{n}{\dfrac{q}{\mathcal{O}(\boldsymbol{b}_i)}}}.
$$
\end{theorem}
\begin{proof}
Let us assume without loss of generality that $\boldsymbol{M}$ is an upper triangular matrix. By the proof of {\cite[Thm $3.6$]{strey2017lattices}}, denoting $\alpha_j$ as the first nonzero component of $\sigma(\boldsymbol{b}_{j})$, we know that 
\begin{equation*}
    {\color{black}\vol \ } \Lambda_{D} = q^{n - k_1} \left(\displaystyle \prod_{j = 1}^{k_1} \alpha_j\right) \displaystyle \prod_{s = 1}^{a} (q^{a - s})^{k_s - k_{s + 1}} = q^{(a - 1)k_1} \left(\displaystyle \prod_{j = 1}^{k_1} \alpha_j\right) q^{n - \sum\limits_{\ell=1}^{a} \hat{k}_{\ell}}.
\end{equation*}
Moreover, we have $\mathcal{O}(\boldsymbol{b}_j) = q/ \alpha_j$ for all $j = 1, 2, \ldots, k_1$. In fact, the third condition assures the existence of integers $m_1, \ldots, m_i$ such that $\sigma(\boldsymbol{b}_j) = \left(0,\ldots, 0, \alpha_j, m_1 \alpha_j, \ldots, m_i \alpha_j\right)$ and, therefore,
\begin{equation*}
	\dfrac{q}{\alpha_j}\sigma(\boldsymbol{b}_j) = \dfrac{q}{\alpha_j} \left(0,\ldots, 0, \alpha_j, m_1 \alpha_j, \ldots, m_i \alpha_j\right) = \left(0,\ldots, 0, q, m_1q, \ldots, m_iq\right) \in q \mathbb{Z}^n.
\end{equation*}
Thus, it follows that $\left({q}/{\alpha_j}\right)\boldsymbol{b}_j = \boldsymbol{0}$ in $\Z_{q}^{n}$, from where $\mathcal{O}(\boldsymbol{b}_j) \leq q/\alpha_i$. Since the other inequality is trivial, we conclude that $\mathcal{O}(\boldsymbol{b}_j) =  {q}/{\alpha_j}$ for all $j = 1, \ldots, k_1$.
Finally, considering $\Lambda_D = \Lambda_A(\mathcal{C})$ as in Theorem \ref{TeoconstDcomoA}, we obtain ({\cite[Prop $3.2$]{costa2017lattices}})
\begin{equation*}
    |\mathcal{C}| = \left| \Lambda_{D} \cap [0,q^a)^n \right| = \dfrac{q^{an}}{ q^{(a-1)k_1}\left(\prod\limits_{j=1}^{k_1} \dfrac{q}{\mathcal{O}(\boldsymbol{b}_j)} \right) \left( q^{n-\sum\limits_{\ell=1}^{a} k_{\ell}}\right)}
	= \dfrac{ q^{\sum\limits_{\ell=1}^{a} k_{\ell}}}{\prod\limits_{j=1}^{n}{\dfrac{q}{\mathcal{O}(\boldsymbol{b}_j)}}},
\end{equation*}
where the second equality occurs since $k_1 = n$.
\end{proof}

Through connections between Constructions D and D' (Theorem \ref{T42}), the previous results on Construction D are adapted next to Construction D'. We note that Corollary \ref{ConstDlinhaLI} is also presented in \cite{barnes1983new, conway2013sphere, sadeghi2006low} for the binary case and in \cite[Thm $2$]{bos1982further}, which deals with an extension of Construction D, called Construction E, over $p$-ary codes ($p$ prime). Furthermore, an equivalent expression to the one presented in Corollary \ref{CorVolumeDlinhaCondTeorema} can be also found in \cite[Thm $3.2.8$]{strey2017construccoes}.

\begin{theorem}\label{num_max_pontosDlinha}
Consider $\mathcal{C}_1^{\perp} \subseteq \cdots \subseteq \mathcal{C}_a^{\perp}\subseteq \Z_q^{n}$ a family of linear codes. Let $\Lambda_{D^{\perp}} = \Lambda_A(\mathcal{C})$ be as in Definition \ref{DefDdual}, where $\mathcal{C}$ is the $q^{a}$-ary code generated by the rows of the matrix $\rho_{q^{a}}(\boldsymbol{G})$ as in Theorem \ref{T2} for this chain of dual codes. Let $\Lambda_{D'}$ be the associated lattice obtained via Construction $D'$. Then, the cardinality of $\mathcal{C}$ and volume of $\Lambda_{D'}$ satisfy
$$ {\color{black}\vol \ }  \Lambda_{D'} = |\mathcal{C}| = \left|\Lambda_{D^{\perp}} \cap [0,q^a)^n\right|  \leq \prod\limits_{s=1}^{a} \left(\prod\limits_{i = r_{a - s} + 1}^{r_{a - s + 1}}{\mathcal{O}(\boldsymbol{h}_i)q^{s-1}}\right) =  \dfrac{q^{\sum\limits_{\ell=1}^{a} r_{\ell}} }{\prod\limits_{i=1}^{r_a}{\dfrac{q}{\mathcal{O}(\boldsymbol{h}_i)}}}.$$
In particular, if $q
= p$ is prime, this upper bound is equivalent to $p^{\sum\limits_{\ell=1}^{a} r_{\ell}}$.
\end{theorem}
\begin{proof}
Theorem \ref{T42} guarantees that $\Lambda_{D'} = q^{a} \Lambda_{D^{\perp}}^{\ast} = q^{a} \Lambda_A^{\ast} (\mathcal{C})$, where $\mathcal{C}$ is the $q^{a}$-ary linear code generated by the rows of the matrix $\rho_{q^{a}}(\boldsymbol{H})$, as in Remark \ref{lemamatrizanivelq}. So, by Corollary \ref{CorVolumeDlinhaviaA}, we have $\text{vol } \Lambda_{D'} = |\mathcal{C}|$ and the result follows since the upper bound is analogous to the one presented in Theorem \ref{num_max_pontos} under appropriate adjustments of notation.
\end{proof}

\begin{corollary} \label{ConstDlinhaLI}
Following the notation above, if $\boldsymbol{h}_1, \ldots,\boldsymbol{h}_{r_a}$ are linearly independent over $\Z_{q}$, then
$${\color{black}\vol \ } \Lambda_{D'} = |\mathcal{C}| = \left| \Lambda_{D^{\perp}}\cap [0,q^a)^n\right| = q^{\sum\limits_{\ell=1}^{a} r_{\ell}}.$$
\end{corollary}

\begin{corollary}\label{CorVolumeDlinhaCondTeorema}
Under the conditions of Theorem \ref{T2} applied to the dual chain {\cite[Thm $2$]{strey2018bounds}} and if $r_a = n$, we have
$${\color{black}\vol \ } \Lambda_{D'} = |\mathcal{C}| = \left| \Lambda_{D^{\perp}} \cap [0,q^a)^n\right| =  \dfrac{q^{\sum\limits_{\ell=1}^{a} r_{\ell}} }{\prod\limits_{i=1}^{n}{\dfrac{q}{\mathcal{O}(\boldsymbol{h}_i)}}}.
$$
\end{corollary}

\begin{remark}
Although the upper bounds presented in this section involve the order of each code generator, we must emphasize that just these orders are not enough to determine the volume of the lattice. For instance, the generators taken in Example \ref{Exemplodo42e24} have the same order in $\Z_{6}^{2}$ and provide lattices with different volumes.
\end{remark}

\subsection{Minimum Distance}

In this subsection, we explore, under specific conditions, the {\color{black} minimum}
 $L_{\mathrm{P}}$-distance of lattices from Constructions D and D' by using the minimum distance of the nested codes or their duals. The results presented here extend to lattices constructed from codes over $\Z_q$ and to $L_{\mathrm{P}}$-distance results from \cite{sadeghi2006low}, which deals with the squared Euclidean minimum distance of lattices from binary codes, and from \cite{strey2017construccoes}, regarding $L_{1}$-distance of lattices obtained from $q$-ary codes.

{\color{black} Besides the Euclidean distance ($\mathrm{P} = 2$), other $L_{\mathrm{P}}$-distances have been considered either theoretically, such as the search for perfect and quasi-perfect codes under $\mathrm{P}$-Lee distance \cite{campello2016perfect, zhang2017perfect, qureshi2016perfect, xu2023almost}, or for applications in Cryptography and communications. Several works in lattice-based cryptography, for instance, analyze the complexity of certain computational problems related to lattices in the $L_{\mathrm{P}}$-norms, such as the closed and the shortest vector problems (CVP and SVP) \cite{peikert2008limits} and the bounded decoding distance (BDD) \cite{bennett2020hardness}. Under a cryptography perspective and aiming at possible applications to error-detection in lattice-based communications, in \cite{chandrasekaran2018local} it has been proposed the study of a computational problem called local testability for membership in lattices, for $L_{\mathrm{P}}$-distances. It should be noted that, in order to obtain nearly matching bounds on the complexity, the authors of \cite{chandrasekaran2018local} focus on families of lattices constructed by Code Formula from a chain of binary Reed-Muller codes closed under the Schur product.

Particularly, the use of $L_{1}$ and $L_{\infty}$-distances plays a role in communications. As mentioned in \cite{etzion2013coding}, the Lee-distance had been considered for BCH codes over fields used in constrained and partial-response channels in \cite{roth1994lee}, for generalized Reed-Muller codes over $\Z_{2^{r}}$, with $r \in \N$, applied to orthogonal frequency-division multiplexing in \cite{schmidt2007complementary}, for general linear codes over $\Z_{p}$, with $p$ prime, in coding for multidimensional bursterror-correction \cite{etzion2009error} and also for error-correction in the rank modulation scheme for flash memories \cite{jiang2010correcting}. Regarding the $L_{\infty}$-distance, in \cite{seethaler2010performance} it is shown that sphere decoding under these distance provides a much smaller computational complexity with a marginal performance loss for  independent and identically distributed (i.i.d) Rayleigh fading multiple-input multiple-output (MIMO) channels.}

We establish a formula for the {\color{black} minimum}  $L_{\mathrm{P}}$-distance of Construction $\overline{D}$, from which we can derive a result for Construction D. This formula will be presented in what follows after the statement of some auxiliary results.

\begin{lemma}\label{lemaexistpontosLp}
Let $\mathcal{C} \subseteq \Z_{q}^{n}$ be a nonzero linear code. Then, for any $1 \leq \mathrm{P} \leq \infty$, we can assert that there exists $\boldsymbol{x}, \boldsymbol{y} \in \mathcal{C}$ such that
\begin{equation*}
    ||\sigma(\boldsymbol{x}) - \sigma(\boldsymbol{y})||_{\mathrm{P}} = d_{\mathrm{P}} (\mathcal{C}).
\end{equation*}
\end{lemma}
\begin{proof}
The proof is straightforward from the fact that the $L_{\mathrm{P}}$-norm in $\Z_{q}^{n}$ is induced by the $L_{\mathrm{P}}$-norm in $\Z^{n}$ for any $1 \leq \mathrm{P} \leq \infty$ {\cite[Prop $2$]{jorge2013q}}.
\end{proof}

\begin{lemma}\label{lemadistanciaLp}
Let $\boldsymbol{z} = (z_1,z_2, \ldots, z_n)$ and $\boldsymbol{r} = (r_1, r_2, \ldots, r_n)$ be vectors of $\Z^n$ such that $0 \leq r_i < q$ for all $i \in \{1,\ldots,n\}$. Then
 $||q\boldsymbol{z}+\boldsymbol{r}||_{\mathrm{P}} \geq ||\boldsymbol{\mu}||_{\mathrm{P}}$, where $\boldsymbol{\mu} := (\mu_1, \ldots, \mu_{n})$ and $\mu_{i} : = \min\left\{q - r_i, r_i\right\}$ for all $i \in \{1,\ldots,n\}$.
\end{lemma}
\begin{proof}
Since the largest negative integer of the form $qz_i+r_i$ is $-q+r_i$
and the smallest positive integer is $r_i$, it follows that $|qz_i+r_i| \geq \min \left\{|-q+r_i|, |r_i|\right\} =  \min \{r_i, q-r_i\}$. Then
$$||q\boldsymbol{z}+\boldsymbol{r}||_{\mathrm{P}} = \left(\sum_{i=1}^{n}|qz_i+r_i|^P\right)^{1/P} \geq \left(\sum_{i=1}^{n}\min \{r_i, q-r_i\}^{P} \right)^{1/P} = ||\boldsymbol{\mu}||_{\mathrm{P}}.$$
 \end{proof}

The next result provides a formula for $L_{\mathrm{P}}$-distances of Construction $\overline{D}$. When $\mathrm{P} = 1$, the Theorem \ref{TeoDistLpDbarra} was proved in {\cite[Thm $3$]{strey2018bounds}} and when $\mathrm{P} = 2$, for a chain of binary codes, in {\cite[Thm $3$]{sadeghi2006low}}.

\begin{theorem}\label{TeoDistLpDbarra}
Let $\{\boldsymbol{0}\} \subsetneq \mathcal{C}_{a} \subseteq \mathcal{C}_{a - 1} \subseteq \cdots \subseteq \mathcal{C}_1 \subseteq \Z_{q}^{n}$ be a family of nested linear codes. Consider the $L_{\mathrm{P}}$-distance, with $1 \leq \mathrm{P} \leq \infty$, and denote the {\color{black} minimum}  $L_{\mathrm{P}}$-distance of  $\mathcal{C}_\ell$ by $d_{\mathrm{P}}(\mathcal{C}_{\ell})$. Then, the {\color{black} minimum}  $L_{\mathrm{P}}$-distance of $\Gamma_{\overline{D}}$ in $\R^n$ satisfies
\begin{equation*}
    d_{\mathrm{P}}(\Gamma_{\overline{D}}) = \min_{1 \leq j \leq a} \left\{q^{a}, q^{a - j} d_{\mathrm{P}}(\mathcal{C}_j)\right\}.
\end{equation*}
\end{theorem}
\begin{proof}
 {\color{black} We use arguments such as the ones in \cite[Thm $3$]{strey2017lattices}. For each $1 \leq \ell \leq a$, by Lemma \ref{lemaexistpontosLp} that there exist $\boldsymbol{x}_{\ell}, \boldsymbol{y}_{\ell} \in \mathcal{C}_{\ell}$ such that $||\sigma(\boldsymbol{x}_{\ell}) - \sigma(\boldsymbol{y}_{\ell})||_{\mathrm{P}} = d_{\mathrm{P}}(\mathcal{C}_{\ell})$. Fix $1 \leq \ell \leq a$. Since $q^{a - \ell} \sigma(\mathcal{C}_{\ell}) \subseteq \Gamma_{\overline{D}}$, it follows that $q^{a - \ell}\sigma(\boldsymbol{x}_{\ell}), q^{a - \ell} \sigma(\boldsymbol{y}_{\ell}) \in \Gamma_{\overline{D}}$. One the one hand, we have
\begin{equation*}
    ||q^{a - \ell} \sigma(\boldsymbol{x}_{\ell}) - q^{a - \ell} \sigma(\boldsymbol{y}_{\ell})||_{\mathrm{P}} = q^{a - \ell} ||\sigma(\boldsymbol{x}_{\ell}) -  \sigma(\boldsymbol{y}_{\ell})||_{\mathrm{P}} = q^{a - \ell} d_{\mathrm{P}}(\mathcal{C}_{\ell}).
\end{equation*}
On the other hand, $q^{a}\Z^{n} \subseteq \Gamma_{\overline{D}}$ so that $d_{\mathrm{P}}(\Gamma_{\overline{D}}) \leq q^{a}$, what proves that
\begin{equation*}
    d_{\mathrm{P}}(\Gamma_{\overline{D}}) \leq \min_{1 \leq j \leq a} \left\{q^{a}, q^{a - j} d_{\mathrm{P}}(\mathcal{C}_j)\right\}.
\end{equation*}
For the other inequality, let $\boldsymbol{x}, \boldsymbol{y} \in \Gamma_{\overline{D}}$ be distinct elements, with $\boldsymbol{x} = q^{s} \boldsymbol{v}$ and $\boldsymbol{y} = q^{k}\boldsymbol{w}$, where $\boldsymbol{v}, \boldsymbol{w} \in \Z^{n}$ and $\boldsymbol{v}, \boldsymbol{w} \not \equiv \boldsymbol{0} \mod q$. Assume $s \geq k$ without loss of generality.
\begin{itemize}
    \item[$(i)$] If $k \geq a$, we have $d_{\mathrm{P}}^{\mathrm{P}}(\boldsymbol{x}, \boldsymbol{y}) = d_{\mathrm{P}}^{\mathrm{P}}(q^{s}\boldsymbol{v}, q^{k}\boldsymbol{w}) = q^{k\mathrm{P}} d_{\mathrm{P}}^{\mathrm{P}}(q^{s - k} \boldsymbol{v}, \boldsymbol{w}) \geq q^{a\mathrm{P}}$ since $\boldsymbol{0} \neq q^{s - k}\boldsymbol{v} - \boldsymbol{w} \in \Z^{n}$.
    \item[$(ii)$] If $0 \leq k \leq a - 1$, then there exist $\boldsymbol{c}_1 \in \mathcal{C}_{1}, \ldots, \boldsymbol{c}_{a - k} \in \mathcal{C}_{a - k}$ and $\boldsymbol{z} \in \Z^{n}$ such that
    \begin{equation*}
        \boldsymbol{y} = q^{a} \boldsymbol{z} + q^{a - 1} \sigma(\boldsymbol{c}_1) + \cdots + q^{k} \sigma(\boldsymbol{c}_{a -  k}),
    \end{equation*}
    which implies
    \begin{equation*}
        \boldsymbol{w} = q^{a - k} \boldsymbol{z} + q^{a - 1 - k} \sigma(\boldsymbol{c}_1) + \cdots + q^{0}\sigma(\boldsymbol{c}_{a - k}).
    \end{equation*}
    Note that $\boldsymbol{w} \! \mod q = q^{0} \sigma(\boldsymbol{c}_{a - k}) \in \sigma(\mathcal{C}_{a - k})$. Denote $\overline{\boldsymbol{w}}:= \rho(\boldsymbol{w})$ and $\overline{\boldsymbol{v}} : = \rho(\boldsymbol{v})$, where $\rho: \Z^{n} \rightarrow \Z_q^{n}$ is the reduction map modulo $q$. Since $\boldsymbol{w} \mod q = q^{0} \sigma(\boldsymbol{c}_{a - k})$, we have $\overline{w} = \boldsymbol{c}_{a - k}\in \mathcal{C}_{a - k}$ and, thus, $q^{s - k} \overline{\boldsymbol{v}} - \overline{\boldsymbol{w}} \in \mathcal{C}_{a - k}$ by using that $s \geq k$. Moreover, due the fact that $\boldsymbol{w} \not \equiv \boldsymbol{0} \mod q$ and $\boldsymbol{w} \neq \boldsymbol{v}$, this is a nonzero vector, which guarantees
     \begin{equation*}
         d_{\mathrm{P}}(\boldsymbol{x}, \boldsymbol{y}) = d_{\mathrm{P}}(q^{s}\boldsymbol{v}, q^{k}\boldsymbol{w}) = q^{k} d_{\mathrm{P}}(q^{s - k}\boldsymbol{v}, \boldsymbol{w}) = q^{k}d_{\mathrm{P}}(q^{s - k}\boldsymbol{v} - \boldsymbol{w}, \boldsymbol{0}) \geq q^{k} d_{\mathrm{P}}(\mathcal{C}_{a - k}),
     \end{equation*}
     where the last inequality follows from
     \begin{equation*}
         |q^{s - k} v_i - w_i| \geq \min \left\{\sigma(q^{s - k} \overline{v}_i - \overline{w}_i), q - \sigma(q^{s - k} \overline{v}_i - \overline{w}_i)\right\},\ \mbox{for each} \ i \in \{1, \ldots, n\}.
     \end{equation*}
\end{itemize}
Finally, we conclude that $d_{\mathrm{P}}(\boldsymbol{x}, \boldsymbol{y}) \geq \displaystyle \min_{1 \leq j \leq a} \left\{q^{a}, q^{a - j} d_{\mathrm{P}}(\mathcal{C}_j)\right\}$, completing the proof.}
\end{proof}

\begin{corollary}\label{cordistLpConstD}
Under the hypothesis of Theorem \ref{TeoDistLpDbarra}, if  $\Lambda_{\overline{D}}$ is the smallest lattice that contains  $\Gamma_{\overline{D}}$, then
\begin{equation*}
    d_{\mathrm{P}} (\Lambda_{\overline{D}}) {\ \color{black} = } \min_{1 \leq j \leq a} \left\{q^{a}, q^{a - j} d_{\mathrm{P}}(\mathcal{C}_j)\right\}.
\end{equation*}
{\color{black} Moreover, it holds that $d_{\mathrm{P}} (\Lambda_{D}) \geq \min_{1 \leq j \leq a} \left\{q^{a}, q^{a - j} d_{\mathrm{P}}(\mathcal{C}_j)\right\}$, with equality if the chain is closed under the zero-one addition. In particular, if
$d_{\mathrm{P}}(\mathcal{C}_{\ell}) \geq q^{\ell}$ for each $1 \leq \ell \leq a$, then $d_{\mathrm{P}}(\Lambda_{D}) = q^{a}$.}
\end{corollary}
\begin{proof}
{\color{black} Since $\Gamma_{\overline{D}} \subseteq \Lambda_{\overline{D}}$ and by Theorem \ref{TeoDistLpDbarra}, we already have
    \begin{equation*}
        d(\Lambda_{\overline{D}}) \leq d_{\mathrm{P}}(\Gamma_{\overline{D}}) = \min_{1 \leq j \leq a} \left\{q^{a}, q^{a - j} d_{\mathrm{P}}(\mathcal{C}_j)\right\}.
    \end{equation*}
    Thus, it is sufficient to prove that
       $ d_{\mathrm{P}}(\Lambda_{\overline{D}}) \geq \displaystyle \min_{1 \leq j \leq a} \left\{q^{a}, q^{a - j} d_{\mathrm{P}}(\mathcal{C}_j)\right\}.$
      Following a similar approach to the proof of Theorem \ref{TeoDistLpDbarra}, let $\boldsymbol{x}, \boldsymbol{y} \in \Lambda_{\overline{D}}$ be distinct elements. By the characterization of the elements of $\Lambda_{\overline{D}}$ \cite[Thm $3.14$]{strey2017lattices}, there exist $\boldsymbol{z}, \boldsymbol{w} \in \Z^{n}$ and $\alpha_{j}^{(i)}, \beta_{j}^{(i)} \in \left\{0, 1, \ldots, q - 1\right\}$ such that
    \begin{equation*}
        \boldsymbol{x} = q^{a} \boldsymbol{z} + \displaystyle \sum_{i = 1}^{a} q^{a - i} \sum_{\boldsymbol{c_{j}} \in \mathcal{C}_i} \alpha_{j}^{(i)} \sigma(\boldsymbol{c}_{j})\  \ \ \mbox{and} \ \ \
        \boldsymbol{y} = q^{a} \boldsymbol{w} + \displaystyle \sum_{i = 1}^{a} q^{a - i} \sum_{\boldsymbol{c_{j}} \in \mathcal{C}_i} \beta_{j}^{(i)} \sigma(\boldsymbol{c}_{j}).
    \end{equation*}
    So $\rho(\boldsymbol{x}), \rho(\boldsymbol{y})  \in \mathcal{C}_{a}$, since
    \begin{equation*}
     \boldsymbol{x} \equiv  \displaystyle \sum_{\boldsymbol{c}_{j} \in \mathcal{C}_{a}} \alpha_{j}^{(a)} \sigma(\boldsymbol{c}_{j}) \mod q \ \ \ \mbox{and} \ \ \ \boldsymbol{y} \equiv  \displaystyle \sum_{\boldsymbol{c}_{j} \in \mathcal{C}_{a}} \beta_{j}^{(a)} \sigma(\boldsymbol{c}_{j}) \mod q.
     \end{equation*}
   Therefore,
     \begin{eqnarray*}
      \min_{1 \leq j \leq a} \left\{q^{a}, q^{a - j} d_{\mathrm{P}}(\mathcal{C}_j)\right\} &\leq &
     d_{\mathrm{P}}(\mathcal{C}_{a})   \leq  d_{\mathrm{P}}\left(\rho(\boldsymbol{x}),\rho(\boldsymbol{y})\right)\\
     & = &  \left[\sum_{i=1}^{n}\left(\min\left\{\sigma\left(\rho(x_i)-\rho(y_i)\right), q - \sigma\left(\rho(x_i)-\rho(y_i)\right)\right\}\right)^{\mathrm{P}} \right]^{1/\mathrm{P}}\\ & \leq &\left(\sum_{i=1}^{n} |x_i-y_i|^{\mathrm{P}} \right)^{1/\mathrm{P}}=  d_{\mathrm{P}}(\boldsymbol{x},\boldsymbol{y}),
     \end{eqnarray*}
     where the last inequality follows from
     \begin{center}
     $\min\left\{\sigma\left(\rho(x_i)-\rho(y_i)\right), q - \sigma\left(\rho(x_i)-\rho(y_i)\right)\right\} \leq |x_i-y_i|$ for all $i \in\{1, \ldots,n\}$
     \end{center}  as in Theorem \ref{TeoDistLpDbarra}. The arguments are analogous for the $L_{\infty}$-distance.  This shows that $\displaystyle \min_{1 \leq j \leq a} \left\{q^{a}, q^{a - j} d_{\mathrm{P}}(\mathcal{C}_j)\right\} \leq d_{\mathrm{P}}(\Lambda_{\overline{D}})$.          The inequality $d_{\mathrm{P}}(\Lambda_{D}) \geq \displaystyle \min_{1 \leq j \leq a} \left\{q^{a}, q^{a - j} d_{\mathrm{P}}(\mathcal{C}_j)\right\}$ follows directly from the fact $\Lambda_{D} \subseteq \Lambda_{\overline{D}}$ {\cite[Thm $3.14$]{strey2017lattices}} and the particular case from Theorem \ref{TeoCadeiafechada}.
     }
\end{proof}

The first part of Corollary \ref{cordistLpConstD}, when $\mathrm{P} = 1$, corresponds to Conjecture $1$ proposed in \cite{strey2018bounds}.

\begin{remark}
We can see that in the proof {\color{black} of Theorem \ref{TeoDistLpDbarra}} the condition of the codes being nested (required for Construction $\overline{D}$) was not used. We could then have considered the more general Construction $C$ for linear codes, which is not approached here, and get an analogous expression.
This generalized the result of \cite{conway2013sphere} and \cite{bollauf2019multilevel} regarding the Euclidean minimum distance of Construction $C$ to $L_{\mathrm{P}}$-distances.
\end{remark}

\begin{example}
Considering the family of codes given in Example \ref{excomordemdiferentedeplasConstD}, observe that both chains are closed under the zero-one addition, since $\mathcal{C}_1 = \Hat{\mathcal{C}}_1 = \Z_{6}^{2}$. Thus, in this case, by Theorem \ref{TeoDistLpDbarra} the $L_{\mathrm{P}}$-distance of the codes $\mathcal{C}_2$ and $\Hat{\mathcal{C}}_2$ determine completely the $L_{\mathrm{P}}$-distance of $\Lambda_{D}$ and $\Hat{\Lambda}_{D}$, respectively. Specifically, for $\mathrm{P} =2$ (Euclidean minimum distance), we obtain $d_{2}(\Lambda_{D}) = \min \left\{36, 6, 1\right\} = 1$ and $d_2 (\Hat{\Lambda}_{D}) = \min \left\{36, 6, \sqrt{5}\right\} = \sqrt{5}$.
\end{example}

\begin{example}
Consider the chain of nested codes $\mathcal{C}_2 \subseteq \mathcal{C}_1 \subseteq \Z_3^{3}$, where $\mathcal{C}_2 =  \langle(1,1,1)\rangle$ and $\mathcal{C}_1 = \langle(1,1,1), (0,0,1)\rangle$. Let $\Lambda_{D}$ be the lattice obtained from Construction $D$ under the generators above. From Theorem \ref{TeoconstDcomoA}, $\Lambda_D = \Lambda_A(\mathcal{C})$, where $\mathcal{C}$ is the $9$-ary code generated by the rows of the matrix
\begin{equation*}
    \boldsymbol{G} = \left[\begin{array}{ccc}
         1 & 1 & 1  \\
         0 & 0 & 3
    \end{array}\right]
\end{equation*}
Thus, by using Hermite Normal Form \cite{costa2017lattices}, we get a generator matrix for $\Lambda_{D}$ given by
\begin{equation*}
    \boldsymbol{M} = \left[\begin{array}{ccc}
        1 & 0 & 0  \\
        1 & 9 & 0\\
        1 & 0 & 3
    \end{array}\right].
\end{equation*}
It is straightforward to see that for $P = 2, 1$ and $\infty$, the Euclidean, Lee and maximum distance of $\Lambda_{D}$ are $d_2(\Lambda_{D}) = \sqrt{3}$, $d_{1}(\Lambda_{D}) = 1$ and $d_{\infty}(\Lambda_D) = 1$, respectively. On the other hand, the minimum distances of the codes $\mathcal{C}_1$ and $\mathcal{C}_2$ in these distances are $d_2(\mathcal{C}_1) = 1$, $d_2(\mathcal{C}_2) = \sqrt{3}$, $d_{1}(\mathcal{C}_1) = 1 = d_{1}(\mathcal{C}_2)$ and $d_{\infty}(\mathcal{C}_1) = 1 = d_{\infty}(\mathcal{C}_2)$. So, for $\mathrm{P} = 1, 2, \infty$, we can verify that $d_{\mathrm{P}}(\Lambda_{D}) = \min \left\{9, 3d_{\mathrm{P}}(\mathcal{C}_1), d_{\mathrm{P}}(\mathcal{C}_2)\right\}$.
Finally, the chain is closed under the zero-one addition since $(1,1,1) \ast (1,1,1) = (0,0,0) \in \mathcal{C}_2$ and $(2,2,2) \ast (2,2,2) = (1,1,1) \in \mathcal{C}_2$. This illustrates Corollary \ref{cordistLpConstD} for the $L_{\mathrm{P}}$-distances, with $\mathrm{P} = 1, 2, \infty$.
\end{example}

In order to provide some bounds for $L_{\mathrm{P}}$-distances of Construction D' for a certain chain of $q$-ary linear codes, we present an auxiliary result. This is one of the so-called Transference's Theorems, which relate some properties of a lattice $\Lambda$ and its dual lattice $\Lambda^{\ast}$. The following version is a consequence of the First and Second Minkowski's Theorems \cite{siegel1989lectures, nguyen2009hermite} and can be also viewed as a particular case of Banaszczyk's Theorem for successive minima {\cite[Thm $2.2$]{banaszczyk1993new}}. 

\begin{theorem}[\cite{banaszczyk1993new, siegel1989lectures}]\label{TeoTransferenciaForte}
{\color{black} Let $\Lambda \subset \R^{n}$ be a full-rank lattice. Then, the
 minimum Euclidean distances of $\Lambda$ and $\Lambda^{\ast}$ satisfy $d_2(\Lambda) \cdot d_2(\Lambda^{\ast}) \leq n$. Another inequality also satisfied is $d_2(\Lambda) \cdot d_2(\Lambda^{\ast}) \leq \gamma_{n}$, where $\gamma_n$ is the Hermite's constant in dimension $n$, that is, $\gamma_n = 4\delta_n^{2/n}$, and $\delta_n$ is the maximum center density for lattices in dimension $n$.}
\end{theorem}

\begin{remark}
It is worth noticing that $\gamma_{n} \leq n/4 + 1$ for all $n$, as shown in \cite{nguyen2009hermite}, what means that the bound with Hermite's constant is more restrictive than with $n$. Unfortunately, the exact value of $\gamma_{n}$ is only known for dimensions $1 \leq n \leq 8$ and $n = 24$.
\end{remark}

Some well-known inequalities involving the $L_{\mathrm{P}}$-distances in $\R^{n}$ allow us to directly derive a consequence from the theorem above.

\begin{corollary}\label{TransfFracaparaLp}
For a full-rank lattice $\Lambda \subset \R^{n}$, we have
\begin{eqnarray}
    d_{\mathrm{P}}(\Lambda) \cdot d_{\mathrm{P}}(\Lambda^{\ast}) &\leq& \gamma_n \leq \dfrac{n}{4} + 1\hspace{0.2cm} \text{ for } 2 < \mathrm{P} \leq \infty; \label{eq1}\\
    d_{\mathrm{P}}(\Lambda) \cdot d_{\mathrm{P}}(\Lambda^{\ast}) &\leq& \left(n^{\frac{1}{p} - \frac{1}{2}}\right)^{2} \gamma_n \leq \left(n^{\frac{1}{p} - \frac{1}{2}}\right)^{2} \left(\dfrac{n}{4} + 1\right) \hspace{0.2cm} \text{ for } 1 \leq \mathrm{P} < 2 \label{eq2}.
\end{eqnarray}
\end{corollary}
\begin{proof}
{\color{black} Recall that from the Holder's inequality for $L_{\mathrm{P}}$-norm \cite{kreyszig1991introductory}, if $2 < \mathrm{P} < \infty$, then $ ||\boldsymbol{x}||_{\mathrm{P}} \leq ||\boldsymbol{x}||_{2} \leq n^{\frac{1}{2} - \frac{1}{\mathrm{P}}}||\boldsymbol{x}||_{\mathrm{P}}$, and for $\mathrm{P} = \infty$, we have $||\boldsymbol{x}||_{\infty} \leq ||\boldsymbol{x}||_2$. Thus, (\ref{eq1}) follows from Theorem \ref{TeoTransferenciaForte}.
For $\mathrm{P} < 2$, we have $||\boldsymbol{x}||_{\mathrm{P}} \leq n^{\frac{1}{\mathrm{P}} - \frac{1}{2}}||\boldsymbol{x}||_{2}$ and, then (\ref{eq2}) holds from Theorem \ref{TeoTransferenciaForte}. }
\end{proof}

\begin{remark}
    Banaszczyk's Theorem \cite{banaszczyk1993new} states that $d_2(\Lambda) \cdot d_2(\Lambda^{\ast}) \leq n$ and this result is considered tight up to a constant. Under this approach, in \cite{miller2019kissing} it is presented another bound
    \begin{equation*}
        d_{2}(\Lambda) \cdot d_{2}(\Lambda^{\ast}) \leq \dfrac{n}{2\pi} + \dfrac{3\sqrt{n}}{\pi}.
    \end{equation*}

   There is a result shown by {\cite[Thm $9.5$]{milnor1973symmetric}} which asserts the existence of a sequence of self-dual lattices that satisfy $d_2(\Lambda) = \Theta(\sqrt{n})$, i.e., bounded below and above by a constant multiple of $\sqrt{n}$. For such lattices, we have $d_2(\Lambda) \cdot d_2(\Lambda^{\ast}) = \Omega(n)$, where $\Omega(n)$ denote a quantity bounded below by a constant multiple of $n$.

    The Corollary \ref{TransfFracaparaLp} states an upper bound for the $L_{\mathrm{P}}$-distance of the dual lattice $\Lambda^{\ast}$ related to the $L_{\mathrm{P}}$-distance of $\Lambda$. Nevertheless, since they are obtained by simply applying inequalities relating to Euclidean minimum distance, they certainly can be improved. In this sense, in \cite{miller2019kissing} it is also proposed an upper bound for the $L_1$-distance {\cite[Thm $3.9$]{miller2019kissing}}, namely
     \begin{equation*}
        d_1(\Lambda) \cdot d_{1}(\Lambda^{\ast}) \leq 0.154264n^{2} \left(1 + 2\pi \sqrt{\dfrac{3}{n}}\right)^{2}.
    \end{equation*}
\end{remark}

To establish a result similar to Corollary \ref{cordistLpConstD} for $L_{\mathrm{P}}$-distances of Construction D' lattice, we use Corollary \ref{TransfFracaparaLp} applied to the chain of dual codes jointly to Theorem \ref{T42}.

\begin{corollary}\label{distLpDlinhaigualdade}
    Let $\mathcal{C}_a \subseteq \cdots \subseteq \mathcal{C}_1 \subsetneq \Z_{q}^{n}$ be a family of nested linear codes. Consider $\Lambda_{D'}$ the lattice obtained from Construction $D'$ and a fixed $L_{\mathrm{P}}$-distance, with $1 \leq \mathrm{P} \leq \infty$. Denote by $d_{\mathrm{P}}(\mathcal{C}_{\ell}^{\perp})$ the {\color{black} minimum} $L_{\mathrm{P}}$-distance of $\mathcal{C}_{\ell}^{\perp}$ for each $1 \leq \ell \leq a$. {\color{black} Thus, }
\begin{equation*}
    d_{\mathrm{P}}(\Lambda_{D'}^{\ast}) { \ \color{black} \geq \ } \min_{1 \leq j \leq a}\left\{1, q^{-j} d_{\mathrm{{P}}}(\mathcal{C}_{a - j + 1}^{\perp})\right\},
\end{equation*}
{\color{black} and the equality holds if the chain of dual codes is closed under the zero-one addition.}
In particular, the $L_{\mathrm{P}}$-distances of $\Lambda_{D'}$ satisfy
\begin{eqnarray*}
   \begin{array}{cccc}
     d_{\mathrm{P}} (\Lambda_{D'}) &\leq& \dfrac{\gamma_n}{\displaystyle \min_{1 \leq j \leq a}\left\{1, q^{-j} d_{\mathrm{{P}}}(\mathcal{C}_{a - j + 1}^{\perp})\right\}} \leq \dfrac{\frac{n}{4} + 1}{\displaystyle \min_{1 \leq j \leq a}\left\{1, q^{-j} d_{\mathrm{{P}}}(\mathcal{C}_{a - j + 1}^{\perp})\right\}} & \text{ for } 2 \leq \mathrm{P} \leq \infty \\
     d_{\mathrm{P}} (\Lambda_{D'}) &\leq& \left(n^{\frac{1}{p} - \frac{1}{2}}\right)^{2} \cdot \dfrac{\gamma_n}{\displaystyle \min_{1 \leq j \leq a}\left\{1, q^{-j} d_{\mathrm{{P}}}(\mathcal{C}_{a - j + 1}^{\perp})\right\}} \leq \dfrac{\frac{n}{4} + 1}{\displaystyle \min_{1 \leq j \leq a}\left\{1, q^{-j} d_{\mathrm{{P}}}(\mathcal{C}_{a - j + 1}^{\perp})\right\}} & \text{ for } 1 < \mathrm{P} < 2,
    \end{array}
\end{eqnarray*}
where $\gamma_n$ is the Hermite's constant in dimension $n$.
\end{corollary}

\begin{proof}
These bounds are a simple consequence of Corollary \ref{cordistLpConstD} and Theorem \ref{TeoTransferenciaForte}, since $\Lambda_{D'} = q^{a}\Lambda_{D^{\perp}}^{\ast}$, as proved in Theorem \ref{T42}. The second part follows directly from Corollary \ref{TransfFracaparaLp}.
\end{proof}

\begin{corollary}\label{corDistLpDlinhaigualdade}
Under the hypothesis of Corollary \ref{distLpDlinhaigualdade}, for $2 \leq \mathrm{P} \leq \infty$, if $d_{\mathrm{P}}(\mathcal{C}_\ell^{\perp}) \geq q^{\ell}$ for each $1 \leq \ell \leq a$, it follows that {\color{black} $d_{\mathrm{P}}(\Lambda_{D'}^{\ast}) = q^{ 1 - a}$ and $d_{\mathrm{P}}(\Lambda_{D'}) \leq \min \left\{\gamma_n, n\right\} \cdot q^{a - 1}$. }
\end{corollary}

\begin{remark}
    Note that the required condition of the corollary above is assumed for binary codes {\color{black}and $\mathrm{P} = 2$} in \cite{conway2013sphere, barnes1983new}.
\end{remark}

The next example shows that the conditions of the chain of nested dual codes in Construction $D'$ being closed under the zero-one addition cannot be omitted in Corollary \ref{distLpDlinhaigualdade} {\color{black} in order to attain the equality}.

\begin{example}
Let $\mathcal{C}_2 \subseteq \mathcal{C}_1 \subseteq \Z_{6}^{2}$ be the family of nested linear codes, where $\mathcal{C}_1^{\perp} = \langle(4,2)\rangle$ and $\mathcal{C}_2^{\perp} = \langle(4,2), (0,1)\rangle$ (as in Example \ref{Exemplodo42e24}).
We know that $\Lambda_{D'}$ and ${\color{black} 36}\Lambda_{D'}^{\ast}$ are generated, respectively, by
\begin{equation*}
    \boldsymbol{M}_{1} = \left[\begin{array}{cc}
        9 & -3  \\
        0 & 6
    \end{array}\right] \hspace{0.2cm} \text{ and } \hspace{0.2cm} \boldsymbol{M}_{2} = \left[\begin{array}{cc}
        4 & 0  \\
        2 & 6
    \end{array}\right].
\end{equation*}
Let $\mathcal{C} \subseteq \Z_{36}^{2}$ be the linear code such that ${\color{black} 36}\Lambda_{D'}^{\ast} = \Lambda_A(\mathcal{C})$. Thus, the Euclidean minimum distance of $\Lambda_{D'}^{\ast}$ follows by the Euclidean minimum distance for $\mathcal{C}$, namely {\color{black} $d_2(\Lambda_{D'}^{\ast}) = \min \left\{2\sqrt{5}/36, 36/36\right\} = \sqrt{5}/18$}. On the other hand, we have
\begin{equation*}
    \min \left\{1, 6^{-1}d_2(\mathcal{C}_2^{\perp}), {\color{black} 6^{-2}} d_{2}(\mathcal{C}_1^{\perp})\right\} = \min \left\{1, 1/6, 6^{-2}(2\sqrt{2})\right\} = {\color{black}\sqrt{2}/18} < \sqrt{5}/18.
\end{equation*}
Note that the chain $\mathcal{C}_2 \subseteq \mathcal{C}_1 \subseteq \Z_6^{2}$ is not closed under the zero-one addition.
\begin{figure}[!ht]
  \centering
  \subcaptionbox{Code $\mathcal{C}_1^{\perp} = \langle(4,2)\rangle \subseteq \mathds{Z}_{6}^{2}$}{\includegraphics[width=.3\textwidth]{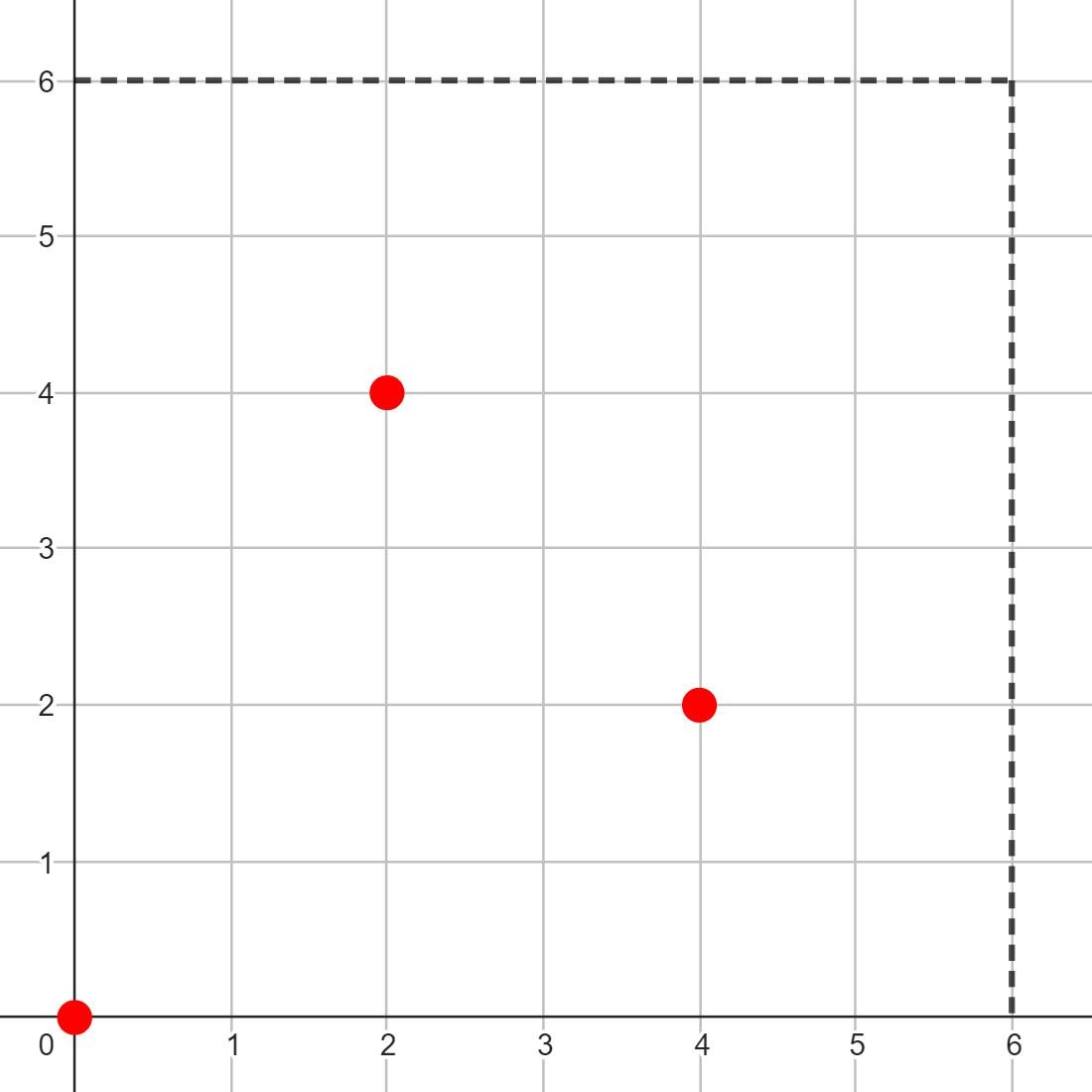}}
  \hspace{0.9cm}
  \subcaptionbox{Code $\mathcal{C}_2^{\perp} = \langle(4,2), (0,1)\rangle \subseteq \mathds{Z}_{6}^{2}$}{\includegraphics[width=.3\textwidth]{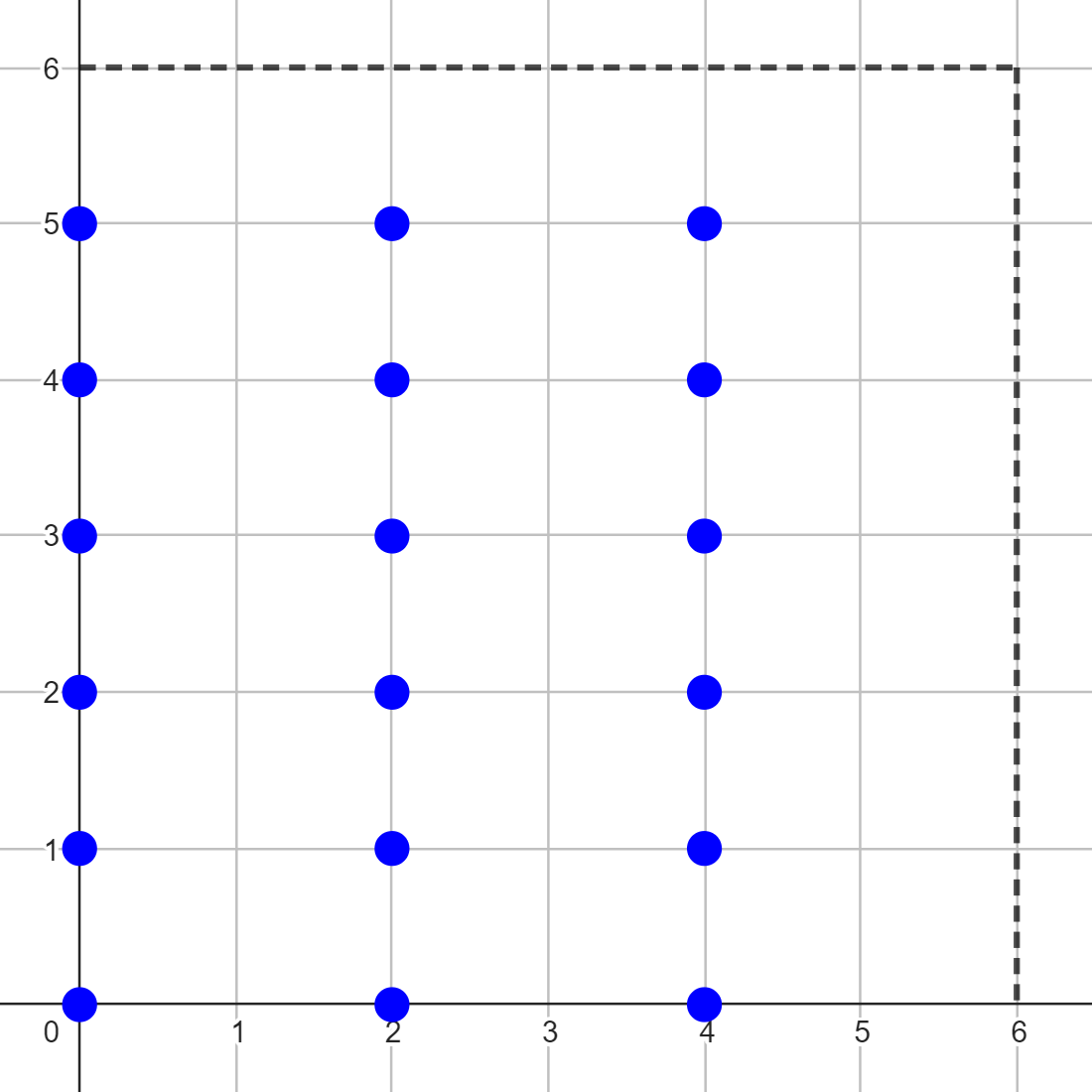}}
\caption{\label{figure12} Dual codes used for Construction D'.}
\end{figure}

\end{example}

{\color{black} For binary codes, it is also possible to obtain {\color{black}a lower bound} for the minimum $L_{\mathrm{P}}$-distance without restrictions under the chain.} These bounds are related to the minimum distance of the original codes and not of their dual codes. We present next some of them, which extend to $L_{\mathrm{P}}$-distances results previously known for Euclidean distance from binary codes {\cite[Thm $3.1$]{sadeghi2006low}} and the ones known for $L_{1}$-distance from $q$-ary linear codes {\cite[Thm $4$]{strey2018bounds}}.

{\color{black}
\begin{lemma}\label{auxiliarystatement}
Let $\{\boldsymbol{0}\} \subsetneq \mathcal{C}_a \subseteq \cdots \subseteq \mathcal{C}_1 \subsetneq \Z_2^{n}$ be a family of nested binary linear codes, $0 \leq r_1 \leq \cdots \leq r_a$ and $\boldsymbol{h}_1, \ldots, \boldsymbol{h}_{r_a} \in \Z_2^{n}$ such that $\mathcal{C}_\ell^{\perp} = \left<\boldsymbol{h}_1, \ldots, \boldsymbol{h}_{r_a}\right>$ for each $\ell = 1, \ldots, n$.  If $\boldsymbol{x} \in \Z^{n}$ has at least one odd coordinate and $\boldsymbol{x} \cdot \sigma(\boldsymbol{h}_{j}) \equiv 0 \mod 2$  for $1 \leq j \leq r_{k}$, then $||\boldsymbol{x}||_{\mathrm{P}} \geq d_{\mathrm{P}}(\mathcal{C}_k)$.
\end{lemma}

\begin{proof}
Let $\boldsymbol{x} = \boldsymbol{c} + 2\boldsymbol{z}$, where $\boldsymbol{z} \in \Z^n$ and $\boldsymbol{c} = (c_1, \ldots, c_n)$, with $c_i = 0$ or $1$. According to the hypothesis, $\boldsymbol{c} \neq \boldsymbol{0}$ and $\boldsymbol{c} \cdot \sigma(\boldsymbol{h}_j) \equiv 0 \mod 2$ for $1 \leq j \leq r_{k}$. Thus, $\rho({\boldsymbol{c}})\in \mathcal{C}_{k}$ and, consequently, $||\boldsymbol{c}||_{\mathrm{P}} \geq d_{\mathrm{P}}(\mathcal{C}_{k})$. On the other hand, by taking $\boldsymbol{\mu} = \boldsymbol{c}$, since $\min \left\{2 - c_i, c_i\right\} = c_i$ for all $i = 1, \ldots, n$, it follows that $||\boldsymbol{x}||_{\mathrm{P}} \geq ||\boldsymbol{c}||_{\mathrm{P}} \geq d_{\mathrm{P}}(\mathcal{C}_{k})$ by Lemma \ref{lemadistanciaLp}.

\end{proof}

\begin{theorem}\label{TeodistLpDlinha}
Let $\{\boldsymbol{0}\}\subsetneq \mathcal{C}_a \subseteq \cdots \subseteq \mathcal{C}_1 \subsetneq \Z_2^{n}$ be a family of nested binary linear codes, $0 \leq r_1 \leq \cdots \leq r_a$ and $\boldsymbol{h}_1, \ldots, \boldsymbol{h}_{r_a} \in \Z_2^{n}$ such that $\mathcal{C}_\ell^{\perp} = \left<\boldsymbol{h}_1, \ldots, \boldsymbol{h}_{r_a}\right>$ for each $\ell = 1, \ldots, n$. Denote by $\Lambda_{D'}$ the lattice obtained via Construction $D'$ using the parameters above and $d_\mathrm{P}(\mathcal{C})$ the $L_{\mathrm{P}}$-distance of $\mathcal{C}_{\ell}$ for each $\ell = 1, \ldots, a$. Then,
\begin{equation*}
   \min_{1 \leq j \leq a} \left\{2^{a}, 2^{a - j} d_{\mathrm{P}}(\mathcal{C}_j)\right\} \leq d_{\mathrm{P}} \left(\Lambda_{D'}\right) \leq 2^{a}.
\end{equation*}
\end{theorem}
\begin{proof}
Let $\boldsymbol{x} \in \Lambda_{D'} \setminus\{\boldsymbol{0}\}$.
Since $\boldsymbol{x}$ is a nonzero vector, we can choose an integer $k \geq 0$ such that $2^{-k}\boldsymbol{x} \in \Z^{n}$, but $2^{-k - 1}\boldsymbol{x} \notin \Z^{n}$.
If $k < a$, then $2^{-k} \boldsymbol{x} \cdot \sigma(\boldsymbol{h}_j) \equiv 0 \mod 2$ for $1 \leq j \leq r_{a - k}$. In fact, we have
\begin{eqnarray*}
    \boldsymbol{x} \in \Lambda_{D'} &\Leftrightarrow& \boldsymbol{x} \cdot \sigma(\boldsymbol{h}_{j}) \equiv 0 \mod 2^{i + 1}\hspace{0.2cm} \text{ for all } 0 \leq i \leq a - 1 \hspace{0.1cm} \text{  and  } \hspace{0.1cm} r_{a - i - 1} < j \leq r_{a - i}\\
    &\Rightarrow& \boldsymbol{x} \cdot \sigma(\boldsymbol{h}_{j}) \equiv 0 \mod 2^{k + 1} \hspace{0.2cm} \text{ for all } k \leq i \leq a - 1 \hspace{0.1cm} \text{  and  } \hspace{0.1cm} r_{a - i - 1} < j \leq r_{a - i}\\
    &\Leftrightarrow& 2^{-k}\boldsymbol{x} \cdot \sigma(\boldsymbol{h}_j) \equiv 0 \mod 2 \hspace{0.2cm} \text{ for } 1\leq j \leq r_{a - k}.
\end{eqnarray*}
Thus, by Lemma \ref{auxiliarystatement}, $||\boldsymbol{x}||_{\mathrm{P}} \geq 2^{k} d_{\mathrm{P}}(\mathcal{C}_{a - k})$. In the other case (that is, if $k \geq a$), we have $\boldsymbol{x} = 2^{a} \boldsymbol{z}$ for some $\boldsymbol{z} \in \Z^n$ and consequently
\begin{equation*}
    ||\boldsymbol{x}||_{\mathrm{P}} =  ||2^{a} \boldsymbol{z}||_{\mathrm{P}} = 2^a ||\boldsymbol{z}||_{\mathrm{P}} \geq 2^{a}.
\end{equation*}
Therefore $$\min_{1 \leq j \leq a} \big\{2^{a}, 2^{a - j} d_{\mathrm{P}}(\mathcal{C}_j)\big\} \leq d_{\mathrm{P}} \left(\Lambda_{D'}\right).$$ To obtain the upper bound for $d_{\mathrm{P}}(\Lambda_{D'})$, it is sufficient to use that $2^{a}\Z^{n} \subseteq \Lambda_{D'}$.
\end{proof}
}

{\color{black}
The next examples illustrate that both bounds in Theorem \ref{TeodistLpDlinha} can be attained.

\begin{example}
Let us consider the family of linear codes given by  $\mathcal{C}_2 \subseteq \mathcal{C}_1 \subseteq \Z_2^{4}$, where $\mathcal{C}_2^{\perp} = \langle(0,0,0,1),(1,1,1,1)\rangle$ and $\mathcal{C}_1^{\perp} = \left<(1,1,1,1)\right>$. We use Theorem \ref{TeodistLpDlinha} to estimate the Euclidean minimum distance of $\Lambda_{D'}$.
Since the Euclidean minimum distance of the codes $\mathcal{C}_1$ and $\mathcal{C}_2$ are $d_2(\mathcal{C}_1) = \sqrt{2} = d_2(\mathcal{C}_2)$, it follows that $\min \left\{4, 2 d_{2}(\mathcal{C}_1), d_{2}(\mathcal{C}_2)\right\} = \sqrt{2}$, which implies $\sqrt{2} \leq d_2(\Lambda_{D'}) \leq 4$ by Theorem \ref{TeodistLpDlinha}. One can show that, in this case, the lower bound is attained. Indeed, by Theorem \ref{propD'viaconstrucaoA} we know that $\Lambda_{D'} = \Lambda_{A}(\mathcal{C}^{\perp})$, where $\mathcal{C}^{\perp}$ is the $4$-ary linear code whose check matrix is
\begin{equation*}
    \boldsymbol{H} = \left[\begin{array}{cccc}
         1 & 1 & 1 & 1   \\
         0 & 0 & 0 & 2
    \end{array}\right].
\end{equation*}
Thus, from Corollary \ref{cormatrizgeradoraD'}, we get a generator matrix for $\Lambda_{D'}$ given by
    \begin{equation*}
    \boldsymbol{M} = \left[\begin{array}{cccc}
        4 & -1 & -1 & -2  \\
        0 & 1 & 0 & 0\\
        0 & 0 & 1 & 0\\
        0 & 0 & 0 & 2
    \end{array}\right],
\end{equation*}
from where the Euclidean minimum distance of $\Lambda_{D'}$ is $\sqrt{2}$.
\end{example}

\begin{example}
    Consider the family of linear codes  $\mathcal{C}_2 \subseteq \mathcal{C}_1 \subseteq \Z_2^{4}$, where $\mathcal{C}_1^{\perp} = \left<(-1,0,1,0),(0,-1,0,1) \right>$ and $\mathcal{C}_2^{\perp} = \langle(-1,0,1,0),(0,-1,0,1),(1,0,0,0),(0,0,0,1)\rangle$.
    By Theorem \ref{propD'viaconstrucaoA} we have that $\Lambda_{D'} = \Lambda_{A}(\mathcal{C}^{\perp})$, where $\mathcal{C}^{\perp}$ is the $4$-ary linear code with check matrix
\begin{equation*}
    \boldsymbol{H} = \left[\begin{array}{cccc}
         -1 & 0 & 1 & 0   \\
         0 & -1 & 0 & 1   \\
         2 & 0 & 0 & 0   \\
         0 & 0 & 0 & 2
    \end{array}\right].
\end{equation*}
So, from Corollary \ref{cormatrizgeradoraD'}, a generator matrix for $\Lambda_{D'}$ is given by
    \begin{equation*}
    \boldsymbol{M} = \left[\begin{array}{cccc}
        4 & 0 & -2 & 0  \\
        0 & 4 & 0 & -2\\
        0 & 0 & 2 & 0\\
        0 & 0 & 0 & 2
    \end{array}\right],
\end{equation*}
from where the Lee minimum distance of $\Lambda_{D'}$ is $4$
which is the upper bound of Theorem \ref{TeodistLpDlinha}.
\end{example}
}

\begin{remark}
The main difficulty of extending the previous result to $L_{\mathrm{P}}$-distances for a chain of $q$-ary linear codes is the Lemma \ref{auxiliarystatement}, which cannot be true under these conditions, unless $P = 1$ or $q = 2$. To the best of our knowledge, it appears that there is no similar result for the general case.
\end{remark}

\subsection{Coding Gain}

As a consequence of the expressions obtained for the minimum Euclidean distance and the volume of the lattices from Constructions D and D', we derive next bounds for the coding gain under specific conditions. For the binary case of Construction D' and a choice of linearly independent generators, Corollary \ref{PropCodinggain}-$(iii)$ is related to what is presented in {\cite[Cor $3.1$]{sadeghi2006low}} with appropriate notation adjustments. The next result follows immediately from the bounds obtained for volume (Theorem \ref{num_max_pontos} and Theorem \ref{num_max_pontosDlinha}, respectively) and Euclidean minimum distance (Corollary \ref{cordistLpConstD} and Corollary \ref{distLpDlinhaigualdade}, respectively) of $\Lambda_{D}$ and $\Lambda_{D'}$.

\begin{corollary}\label{PropCodinggain}
   Let $\Lambda_{D}$ (respectively, $\Lambda_{D'}$) be the lattice obtained via Construction $D$ (respectively, Construction $D'$) following the usual notation and choice of parameters. For a chain  $\{\boldsymbol{0}\}\neq \mathcal{C}_a \subseteq \cdots \subseteq \mathcal{C}_1 \subseteq \Z_{q}^{n}$ (respectively, a associated dual chain $\{\boldsymbol{0}\} \neq \mathcal{C}_1^{\perp} \subseteq \cdots \subseteq \mathcal{C}_a^{\perp} \subseteq \Z_q^{n}$), we have the following results:
   {\color{black}
   \begin{itemize}
       \item[$(i)$] Under the conditions of Theorem \ref{T2} and if $k_1 = n$, we have
       \begin{equation*}
           \gamma(\Lambda_{D}) \geq \dfrac{ \displaystyle \min_{1 \leq j \leq a} \left\{q^{2a}, q^{2(a - j)} d_{\mathrm{P}}^{2}(\mathcal{C}_{j})\right\}}{\left((q^{2})^{a - \sum\limits_{\ell=1}^{a} \frac{k_{\ell}}{n}} \left(\prod\limits_{i=1}^{k_1}{\dfrac{q}{\mathcal{O}(\boldsymbol{b}_i)}} \right)^{2/n}\right)}.
       \end{equation*}

       \item[$(ii)$] If the chain is closed under zero-one addition, then
       \begin{equation*}
           \gamma(\Lambda_{D}) \leq \dfrac{\displaystyle \min_{1 \leq j \leq a} \left\{q^{2a}, q^{2(a - j)} d_{\mathrm{P}}(\mathcal{C}_j)\right\}}{\left((q^{2})^{a - \sum\limits_{\ell=1}^{a} \frac{k_{\ell}}{n}} \left(\prod\limits_{i=1}^{k_1}{\dfrac{q}{\mathcal{O}(\boldsymbol{b}_i)}} \right)^{2/n}\right)}.
       \end{equation*}
       In particular, if the conditions of $(i)$ and $(ii)$ are satisfied the equality holds.
       \item[$(iii)$] Under the conditions of Theorem \ref{T2} for the dual chain and if $r_a = n$, we have
\begin{equation*}
    \gamma(\Lambda_{D'}) \leq \dfrac{ \gamma_{n}q^{2a} \cdot \left(\prod\limits_{i=1}^{r_a}{\dfrac{q}{\mathcal{O}(\boldsymbol{h}_i)}}\right)^{2/n}}{(q^{2})^{\sum\limits_{\ell=1}^{a} \frac{r_{\ell}}{n}} \displaystyle \min_{1 \leq j \leq a} \left\{q^{2a}, q^{2(a - j)} d_{\mathrm{P}}(\mathcal{C}_j^{\perp})\right\}} \leq \dfrac{ \left(\frac{n}{4} + 1\right)q^{2a} \cdot \left(\prod\limits_{i=1}^{r_a}{\dfrac{q}{\mathcal{O}(\boldsymbol{h}_i)}}\right)^{2/n}}{(q^{2})^{\sum\limits_{\ell=1}^{a} \frac{r_{\ell}}{n}} \displaystyle \min_{1 \leq j \leq a} \left\{q^{2a}, q^{2(a - j)} d_{\mathrm{P}}(\mathcal{C}_j^{\perp})\right\}}.
\end{equation*}
\end{itemize}}

\end{corollary}

\begin{remark}
    The coding gain and the center density of a lattice $\Lambda$ are related by $\delta(\Lambda) = 2^{-n} \gamma(\Lambda)^{n/2}$, from what similar bounds for the center density with respect to the Euclidean distance are given.
\end{remark}

{\color{black} We emphasize that, under the conditions of Corollary \ref{PropCodinggain}-$(i)$ for Construction $D$, it is possible to obtain good lattices in low dimensions with respect to packing density. This is the case, for instance, of the constructions of lattices via Construction $D$ from a family of linear codes over $\Z_{4}$ equivalent to $E_{8}$, $BW_{16}$ and $\Lambda_{24}$, as presented in \cite{strey2017construccoes, strey2017lattices}.

Regarding the upper bound given in Corollary \ref{PropCodinggain}-$(ii)$, it is interesting to note that some chains of generalized linear Reed-Muller over $\Z_{q}$, where $q$ is a prime power \cite{bhaintwal2010generalized}, are closed under the zero-one addition. In order to verify this, let us denote the $r$-th generalized Reed-Muller code of length $2^{m}$, $RM_{\Z_{q}}(r,m)$, where $0 \leq r \leq m(p - 1)$ and $q = p^{s}$, with $p$ prime. Using the concept of generalized Boolean functions, $RM_{q}(m,r)$ is defined as the linear code over $\Z_{q}$ generated by the set of all monomials of order at most $r$ in $m$ variables. Equivalently, $RM_{q}(m,r)$ is obtained from all the $\Z_q$-linear combinations of the rows of the generator matrix for the classical binary Reed-Muller
codes \cite{paterson2000efficient}. The next result presents a chain of  generalized Reed-Muller codes that is closed under the zero-one addition, as well as in the binary case \cite{conway2013sphere, kositwattanarerk2014connections}.

\begin{theorem}
Under the above notation, the following chain is closed under the zero-one addition
\begin{equation*}
    RM_{\Z_{q}}(m, 2^{0}) \subseteq RM_{\Z_{q}}(m,2) \subseteq RM_{\Z_{q}}(m,2^{2}) \subseteq \cdots \subseteq RM_{\Z_{q}}(m, 2^{\log_{2} 2^{m}}) = \Z_{q}^{2^{m}}.
\end{equation*}
\end{theorem}
\begin{proof}
In fact, note that the sum of two monomials with a degree at most than $r$ results in a monomial of degree at most $2r$, and the zero-one addition can not increase the order of a monomial. Since in this case $r$ is a power of $2$, follows that the previous chain is closed under the zero-one addition.
\end{proof}

We point out that the class of generalized Reed-Muller codes have good properties for decoding purposes, as shown in \cite{paterson2000efficient, schmidt2007complementary}. Certain special families of quaternary linear Reed-Muller codes have been attracted attention due to their relation with the associated binary linear Reed-Muller codes obtained from Gray map \cite{borges2005quaternary, pujol2008construction, pernas2011classification}. Also, it is known that a family of binary Reed-Muller codes allows constructing Barnes-Wall lattices from Construction $D$ \cite{forney1988coset}.}

\section{Coding and Decoding of Construction \texorpdfstring{$D'$}{D'} for certain \texorpdfstring{$q$}{q}-ary codes}

Several methods for encoding and multistage decoding for binary Constructions D and D' have been proposed recently, with approaches using re-encoding \cite{matsumine2018construction, zhou2022construction, vem2014multilevel, zhou2021encoding}, by computing cosets \cite{silva2020multilevel} and by applying a min-sum algorithm at each level of decoding, as proposed in \cite{sadeghi2006low, sadeghi2010, sadeghi2013performance}. In this paper, we focus on multistage decoding with re-encoding following the approach proposed by \cite{zhou2022construction}. We extend some results to a class of lattices obtained by Construction D' from nested $q$-ary linear codes. The original method performs re-encoding via the check matrix in the sense of \cite{zhou2022construction, sommer2009shaping}, that is, as an inverse of a generator matrix for the lattice. In what follows, the notation of \cite{zhou2021encoding} is applied to our approach.

\subsection{Encoding Method B}

In \cite{zhou2022construction}, two equivalent encoding methods are given, called Encoding Method A and Encoding Method B. The first requires that the check matrix is in the ALT form and can be efficient when the matrix is sparse \cite{zhou2021encoding}. The second one requires that the generators are linearly independent over $\Z_2$ and the check matrix is square. We focus here on Encoding Method B.

Following the established notation and adopting an approach completely analogous to \cite{zhou2022construction}, let $\Lambda_{D'}$ be the lattice obtained via Construction $D'$ from a chain of  $q$-ary linear codes $\mathcal{C}_a \subseteq \cdots \subseteq \mathcal{C}_1 \subseteq \Z_{q}^{n} =: \mathcal{C}_0$ similarly to Definition \ref{defiD'eleonesio}, that is,
\begin{equation*}
    \Lambda_{D'} = \big\{\boldsymbol{x} \in \Z^{n}: \boldsymbol{H} \boldsymbol{x} \equiv \boldsymbol{0} \mod q^{a}\big\},
\end{equation*}
and assume that the linearly independent generators $\boldsymbol{h}_1, \ldots, \boldsymbol{h}_{r_a}$ over $\Z_q$ are completed with $\boldsymbol{h}_{r_a}, \ldots, \boldsymbol{h}_{n}$ in such a way that $\boldsymbol{H}_{a}$ is invertible over $\Z_{q}$. Only in this section, for simplicity, we consider that $r_0$ denotes the number of generators for $\mathcal{C}_0$ obtained from the code generators of the underlying codes. Let $\boldsymbol{x} \in \Lambda_{D'}$ be a lattice vector, denote $\boldsymbol{H} \boldsymbol{x} = q^{a}\boldsymbol{b}$, where $\boldsymbol{b} \in \Z^{n}$, and write
\begin{eqnarray}
   \begin{array}{crl}\label{escritacodifdebgeral}
         b_{j} = & z_j &\text{ for } 1 \leq j \leq r_1 ;\\
         b_j =& u_{1j} + qz_j &\text{ for } r_1 < j \leq r_2 ;\\
        b_j =& u_{2j} + qu_{1j} + q^{2}z_j& \text{ for } r_2 < j \leq r_3 ;\\
        \vdots &  \vdots & \vdots \\
         b_{j} =& u_{(a - 2)j} + qu_{(a - 3)j} + \cdots + q^{a - 3}u_{1j} + q^{a - 2}z_j &\text{ for } r_{a - 2} < j \leq r_{a - 1};\\
         b_{j} =& u_{(a - 1)j} + qu_{(a - 2)j} + q^{2}u_{(a - 3)j} + \cdots + q^{a - 2}u_{1j} + q^{a - 1}z_j &\text{ for } r_{a - 1} < j \leq r_{a};\\
         b_{j} =& u_{aj} + qu_{(a - 1)j} + q^{2}u_{(a - 2)j} + q^{3}u_{(a - 3)j} + \cdots + q^{a - 1}u_{1j} + q^{a}z_j &\text{ for } r_{a} < j \leq n,
    \end{array}
\end{eqnarray}
where $\boldsymbol{u}_i = (u_{i1}, \ldots, u_{i(n - r_i)})^{T} \in \Z^{(n - r_i)}$, $u_{ij} \in \sigma^{\ast}(\Z_{q})$ for each $i = 0, \ldots, a - 1$ and $j = 1, \ldots, n$, and $\boldsymbol{z} \in \Z^{n}$. Here, $\sigma^{\ast}(\Z_{q})$ denote a choice of centralized representatives class, i.e.,
\begin{eqnarray*}
    \sigma^{\ast}(\Z_{q}) &:=& \left\{-\dfrac{q - 1}{2}, - \dfrac{q - 3}{2}, \ldots, 0, \ldots, \dfrac{q - 3}{2}, \dfrac{q - 1}{2}\right\} \hspace{0.2cm} \text{ if } q \text{ is odd; }\\
    \sigma^{\ast}(\Z_{q}) &:=& \left\{-\dfrac{q}{2}, - \dfrac{q - 2}{2}, \ldots, 0, \ldots, \dfrac{q - 4}{2}, \dfrac{q - 2}{2}\right\} \hspace{0.2cm} \text{ if } q \text{ is even }.
\end{eqnarray*}

Similarly to \cite{zhou2022construction}, we consider $\Tilde{\boldsymbol{u}}'_{i}$ as obtained from $\boldsymbol{u}_{i} = (u_{i(r_i + 1)}, \ldots, u_{in})^{T}$ with adjunction of zero coordinates. To preserve the adopted notation, the null coordinates will be added at the beginning, as follows
\begin{equation*}
    \Tilde{\boldsymbol{u}}'_{i} := (\underbrace{0, \ldots, 0}_{r_i},u_{i(r_i + 1)}, \ldots, u_{in})^{T} \in \Z_{q}^{n}.
\end{equation*}

\begin{lemma}\label{lemaEscritadebGeral}
Let $\boldsymbol{b} = (b_1, \ldots, b_n)^{T} \in \Z^{n}$ as in (\ref{escritacodifdebgeral}). Considering the vectors $\Tilde{\boldsymbol{u}}'_{i}$ as before, we have
\begin{equation*}
    \boldsymbol{b} = \boldsymbol{D} (q^{-a} \Tilde{\boldsymbol{u}}'_{0} + q^{-(a - 1)}\Tilde{\boldsymbol{u}}'_{1} + q^{-(a - 2)} \Tilde{\boldsymbol{u}}'_{2} + \cdots + q^{-1}\Tilde{\boldsymbol{u}}'_{a - 1} + \boldsymbol{z}),
\end{equation*}
where $\boldsymbol{D}$ is the diagonal matrix presented in Remark \ref{lemamatrizanivelq}.
\end{lemma}
\begin{proof}
Denote
\begin{equation*}
    \boldsymbol{c} := \boldsymbol{D} (q^{-a}\Tilde{\boldsymbol{u}}'_{0} + q^{-(a - 1)}\Tilde{\boldsymbol{u}}'_{1} + q^{-(a - 2)} \Tilde{\boldsymbol{u}}'_{2} + \cdots + q^{-1}\Tilde{\boldsymbol{u}}'_{a - 1} + \boldsymbol{z}) =: \boldsymbol{D} \Tilde{\boldsymbol{c}}.
\end{equation*}
We can write
    $\Tilde{\boldsymbol{c}} =
    q^{-a}(0, \ldots, 0, u_{0(r_0 + 1)}, \ldots, u_{1(n -r_0)})^{T} + \cdots + q^{- 1} (0, \ldots, 0, u_{(a - 1)(r_{a - 1})}, \ldots, u_{(a - 1)n})^{T}  + (z_1, \ldots, z_n)^{T}.$
Now, multiplying each term by the matrix $\boldsymbol{D}$, we get $\boldsymbol{D}\Tilde{\boldsymbol{c}} = \boldsymbol{b}$, as described in (\ref{escritacodifdebgeral}).
\end{proof}

In a natural extension of the binary case, observe that a vector $\boldsymbol{x} \in \Lambda_{D'}$ can be written as
\begin{equation}\label{eqdecompositionfordecoding}
    \boldsymbol{x} = \boldsymbol{x}_0 + q\boldsymbol{x}_1 + \cdots + q^{a}\boldsymbol{x}_{a} = \displaystyle \sum_{i = 0}^{a} q^{i} \boldsymbol{x}_i,
\end{equation}
where the components $\boldsymbol{x}_{i} \in \Z^{n}$ depend on $\Tilde{\boldsymbol{u}}'_{i}$ for $i = 0, \ldots, a - 1$.
Thus, if we denote $\boldsymbol{H} \boldsymbol{x} = q^{a}\boldsymbol{b} = : \Tilde{\boldsymbol{b}}$, it follows
\begin{equation*}
    \Tilde{\boldsymbol{b}} = \boldsymbol{D}(\Tilde{\boldsymbol{u}}'_{0} + q \Tilde{\boldsymbol{u}}'_{1} + \cdots + q^{a - 1} \Tilde{\boldsymbol{u}}'_{a - 1} + q^{a}\boldsymbol{z}).
\end{equation*}
Under these conditions, since  $\Tilde{\boldsymbol{b}} \in q^{a}\Z^{n}$ is known, we can calculate the components of $\boldsymbol{x}$ by using the relations below
\begin{eqnarray*} 
    \boldsymbol{H}_{a} \boldsymbol{x}_{i} &=& \Tilde{\boldsymbol{u}}'_{i}\hspace{0.2cm} \text{ for each } i = 0, \ldots, a - 1,\\
    \boldsymbol{H}_{a} \boldsymbol{x}_{a} &=& \boldsymbol{z}.
\end{eqnarray*}

\subsection{Decoding of Construction D'}

A natural extension of the decoding approach developed in \cite{zhou2022construction} for Construction D' of a family of $q$-ary linear codes, under the previous conditions, is described next. We state a generalization to Proposition $2$ of \cite{zhou2022construction} for $q$-ary linear codes, when $\boldsymbol{H}_{a}$ is invertible over $\Z_q$ and the proof is analogous to the binary case, with the appropriate notation adjustments.

\begin{proposition}
For Construction D', the lattice component $\boldsymbol{x}_i$ is congruent modulo $q$ to a codeword $\Tilde{\boldsymbol{x}}_i \in \mathcal{C}_i$, for each $i = 0, \ldots, a - 1$.
\end{proposition}
\begin{proof}
Denote $\Tilde{\boldsymbol{x}}_i : = \boldsymbol{x}_i \mod q$ for each $i = 0, \ldots, a - 1$. By the definition of the lattice components, we know that $\boldsymbol{x}_i$ satisfies $\boldsymbol{H}_a \boldsymbol{x}_i = \Tilde{\boldsymbol{u}}'_i$, where the first $r_i$ components of $\Tilde{\boldsymbol{u}}'_{i}$ are zero. Thus, it results that $\boldsymbol{H}_{a,i} \Tilde{\boldsymbol{x}}_i \equiv 0 \mod q$, where $\rho(\boldsymbol{H}_{a,i})$ is the check matrix of $\mathcal{C}_i$ (i.e., corresponds to the first $r_i$ rows of the matrix $\boldsymbol{H}_a$). Equivalently, we can write $\sigma(\boldsymbol{h}_j) \cdot \boldsymbol{x}_i \equiv 0 \mod q$, i.e., $\boldsymbol{h}_j \cdot \Tilde{\boldsymbol{x}}_i = 0$ in $\Z_{q}^{n}$, for $1 \leq j \leq r_i$. Therefore, by using the definition of $\mathcal{C}_i$ by its check matrix, we conclude $\Tilde{\boldsymbol{x}}_i \in \mathcal{C}_i$.
\end{proof}

Under these conditions, the decoding algorithm of Construction D' for a chain of $q$-ary linear codes is essentially the algorithm proposed by \cite{zhou2022construction}. Since the Construction D' was defined for $q$-ary codes from an arbitrary set of tuples in $\Z_q^{n}$, one point that we should be careful about is requiring that $\boldsymbol{h}_1, \ldots, \boldsymbol{h}_a$ are linearly independent over $\Z_q$. This hypothesis is crucial for certain stages of the algorithm (specifically, lines $4$ and $9$) and guarantees that the entries of $\boldsymbol{H}_{a}$ are not zero divisors of $\Z_{q}$, which in practice would weak the distance spectrum of the codes and inhibit the completion convergence of the decoders \cite{sridhara2005ldpc}.

In what follows, {\color{black} a message is a lattice point $\boldsymbol{x} \in \Lambda_{D'}$ and the channel output is $\boldsymbol{y} = \boldsymbol{x} + \boldsymbol{w}$, where $\boldsymbol{w}$ is the noise.} Also, $\overline{\boldsymbol{u}}$ denotes the vector with all coordinates equal to $\lfloor q/2\rfloor$ (integer part) and $\mod_q(\boldsymbol{y}_i + \overline{\boldsymbol{u}})$ denotes the vector obtained by reducing modulo $q$. The decoder $Dec_i$ calculates a codeword $\Hat{\Tilde{\boldsymbol{x}}}_i$ closest to $\boldsymbol{y}'_{i}$ in the $q$-ary linear code $\mathcal{C}_i$, which is an estimate of $\Tilde{\boldsymbol{x}}_i$.

{\color{black} In a theoretical view, the next theorem provides a necessary condition for the decoders $Dec_{i}$ to find the closest $n$-tuple over $\Z_q$ to a received vector $\boldsymbol{y}'_{i}$ over an additive white Gaussian noise (AWGN). Similar results are proposed for decoding binary turbo Construction $D'$ lattices \cite{sakzad2010construction} and decoding the Leech lattice \cite{forney1989bounded}.

\begin{theorem}
Given an $n$-uple $\boldsymbol{y}'_{i}$, if there exists a point $\Tilde{\boldsymbol{x}}_i \in \mathcal{C}_i$ such that $||\boldsymbol{y}'_{i} - \Tilde{\boldsymbol{x}}_i||_{2} \leq d_{2}(\Lambda_{A}(\mathcal{C}_i))/2$, then at each step in the line $8$ the algorithm decoders $\boldsymbol{y}'_{i}$ to $\Tilde{\boldsymbol{x}}_i$, i.e., $\Hat{\Tilde{\boldsymbol{x}}}_i = \Tilde{\boldsymbol{x}}_{i}$. In particular, if the noise $\boldsymbol{w}$ satisfies
\begin{equation*}
 \Big{|}\Big{|}\mod_{q}^{\ast}\left(\dfrac{\boldsymbol{w}}{q^{i}}\right)\Big{|}\Big{|}_{2} \leq \dfrac{1}{2},
\end{equation*}
where $\mod^{\ast}$ denotes the ``triangular function'', that is, $mod^{\ast}_{q}(\boldsymbol{w}) : = |\mod_q(\boldsymbol{w} + \overline{\boldsymbol{u}}) - \overline{\boldsymbol{u}}|$, then the algorithm decoders $\boldsymbol{y}'_{i}$ to $\Tilde{\boldsymbol{x}_{i}}$.
\end{theorem}
\begin{proof}
Based on the geometric uniformity of lattices, it suffices to consider $\boldsymbol{x} = \boldsymbol{0}$ and, hence, under the notation of decomposition (\ref{eqdecompositionfordecoding}), $\boldsymbol{x}_i = \boldsymbol{0}$ for each $i = 0, 1, \ldots, a$. Let us say that there is an error at step $k$ if $\Hat{\Tilde{\boldsymbol{x}}}_{k} \neq \boldsymbol{0}$.

Assume that there have been no errors at former steps $0 \leq k < i$. Since $\Hat{\boldsymbol{x}}_{k} = \boldsymbol{0}$ for $k = 0, \ldots, i - 1$, in the $i$-th step we have $\boldsymbol{y}_{i} = \boldsymbol{w}/q^{i}$ and, then,
\begin{equation*}
    \boldsymbol{y}'_{i} = \Big{|}\mod^{\ast}_{q}\left(\dfrac{\boldsymbol{y}_{i - 1}}{q} + \overline{\boldsymbol{u}}\right) - \overline{\boldsymbol{u}}\Big{|} = \Big{|} \mod_{q}^{\ast}\left(\dfrac{\boldsymbol{w}}{q^{i}} + \overline{\boldsymbol{u}}\right) - \overline{\boldsymbol{u}}\Big{|} = \mod_{q}^{\ast} \left(\dfrac{\boldsymbol{w}}{q^{i}}\right).
\end{equation*}

Under the hypothesis $\Big{|}\Big{|}\mod^{\ast}_{q}\left(\dfrac{\boldsymbol{w}}{q^{i}}\right)\Big{|}\Big{|}_{2} \leq \dfrac{d_{2}(\mathcal{C}_i)}{2}$, the vector $\boldsymbol{y}'_{i}$ is in the sphere packing of $\Lambda_{A}(\mathcal{C}_i)$ and hence no errors occur. It is sufficient to note that $d_{2}(\mathcal{C}_0) \leq \cdots \leq d_{2}(\mathcal{C}_a)$ to complete the proof.

\end{proof}}

\vspace{0.2cm}

\begin{algorithm}[H]
	\LinesNumbered
	\SetAlgoLined
	\KwIn{finite ring $\Z_q$, received message with noisy $\boldsymbol{y}$, full-rank matrix $\boldsymbol{H}_a$.}
	\KwOut{estimated lattice point $\Hat{\boldsymbol{x}}\in\Lambda_{D'}$.}
	$\boldsymbol{y}_0 \leftarrow \boldsymbol{y}$;\\
        $\boldsymbol{y}'_{0} \leftarrow |\mod_q(\boldsymbol{y}_0 + \overline{\boldsymbol{u}}) - \overline{\boldsymbol{u}}|$;\\
        $\Hat{\Tilde{\boldsymbol{x}}}_{0} \leftarrow Dec_{0}(\boldsymbol{y}'_{0})$;\\
        $\hat{\Tilde{\boldsymbol{u}}}'_{1} \leftarrow \boldsymbol{H}_a \Hat{\Tilde{\boldsymbol{x}}}_0 \mod q$, then solve $\boldsymbol{H}_a \Hat{\boldsymbol{x}}_0 = \sigma(\hat{\Tilde{\boldsymbol{u}}}'_{1})$;\\
	\For{$1, 2, \ldots, a - 1$}{
		$\boldsymbol{y}_ i \leftarrow (\boldsymbol{y}_{i - 1} - \Hat{\boldsymbol{x}}_{i - 1})/q$;\\
         $\boldsymbol{y}'_{i} \leftarrow |\mod_q(\boldsymbol{y}_i + \overline{\boldsymbol{u}}) - \overline{\boldsymbol{u}}|$;\\
         $\Hat{\Tilde{\boldsymbol{x}}}_i \leftarrow Dec_{i}(\boldsymbol{y}'_{i})$;\\
         $\Hat{\Tilde{\boldsymbol{u}}}'_{i + 1} \leftarrow \boldsymbol{H}_a \Hat{\Tilde{\boldsymbol{x}}}_i \mod q$, then solve $\boldsymbol{H}_a \Hat{\boldsymbol{x}}_i = \sigma(\hat{\Tilde{\boldsymbol{u}}}'_{i})$
	}
	$\boldsymbol{y}_{a} \leftarrow (\boldsymbol{y}_{a - 1} - \Hat{\boldsymbol{x}}_{a - 1})/q$;\\
	$\Hat{\boldsymbol{x}}_a \leftarrow \lfloor{\boldsymbol{y}_a}\rceil$;\\
	$\Hat{\boldsymbol{x}} \leftarrow \Hat{\boldsymbol{x}}_{0} + q \Hat{\boldsymbol{x}}_{1} + \cdots + q^{a - 1}\Hat{\boldsymbol{x}}_{a - 1} + q^{a}\Hat{\boldsymbol{x}}_a$.
	\label{algo1}
	\caption{Decoding Construction D' Lattices}
\end{algorithm}

\vspace{0.2cm}

Multilevel lattice constructions based on codes have the promise of attain a manageable decoding complexity. On the other hand, it is worth emphasizing that the decoding algorithms for a code $\mathcal{C}_i$ in each interaction must be an efficient one. In a practical view, some nearest-neighbor lattice decoding schemes may not be feasible to implement even for $p$-ary linear codes, where $p$ is prime \cite{silva2020multilevel}. Motivated by the construction of lattices with good performance over AWGN channels and a manageable decoding complexity, several works focus on certain families of nested codes for Construction $D$ and $D'$ over a field. Among these, there are designs and decoding processes for lattices based on $p$-ary linear low-density parity-check (LDPC) codes, which can be decoded by belief propagation (BP) or min-sum algorithms \cite{sadeghi2006low}, generalized low density (GLD) codes, by BP decoding \cite{di2015non} and turbo codes, by using soft-input soft-output (SISO) and soft-input hard-output (SIHO) decoding algorithms \cite{sakzad2010construction}.

Although those classes of codes allow generalizations to codes over $\Z_{q}$, the ring size, as in $\Z_{p}$, with $p$ prime, can affect the decoding complexity. Especially for algorithms based on belief decoding, this leads most works to consider codes over rings that admit a fast Fourier transform, which can provide a reasonable decoding complexity \cite{goupil2006belief}. These classes include nested codes over $\Z_{2k}$ and $\Z_{p^{r}}$, with $p$ prime, with good decoding properties, such as LDPC codes \cite{armand2006decoding, ferrari2012low}, turbo codes \cite{reid2005rate}, low-rank parity-check codes (LRPC) \cite{renner2020low}, BCH, Reed-Solomon \cite{interlando1997decoding}, generalized Reed-Muller codes \cite{paterson2000efficient} over $\Z_{2k}$, and Reed-Solomon codes over $\Z_{p^{r}}$ \cite{mook2021lattice}. It is expected that for families of codes belonging to these classes, $Dec_{i}$ chosen as the proper mentioned decoder
to be applied at each level $i$ in the Decoding Construction $D'$ lattice algorithm above could provide efficient decoding.

\section{Conclusion} \label{SecConclusao}

The volume and $L_{\mathrm{P}}$-distances of Construction $D$ and $D'$ are investigated here considering generator matrices for these constructions. An upper bound for the volume by using a generator and a check matrix, respectively, is presented. We also provide an expression for $L_{\mathrm{P}}$-distances of Construction $\overline{D}$ in terms of the minimum distance of underlying codes and derive some bounds for $L_{\mathrm{P}}$-distances of Construction D and D', under certain conditions. In addition, it is established bounds for the coding gain and a sufficient condition for achieving it. A multistage decoding method with re-encoding applied to Construction D' from $q$-ary linear codes under specific conditions is adapted from \cite{zhou2022construction}. Further work in the directions presented here includes the discussion of efficient decoding for Construction D' for $q$-ary lattices considered in a more general context and possible dependency of the decoding complexity and coding gain on certain lattice parameters, such as the choice of generators.

\section*{Acknowledgements}

The authors wish to thank the Editors of this Special Issue and the referees for their important comments which have mindfully improved the original manuscript. They also thank Juliana G. F. Souza for very fruitful discussions. This work is partially supported by Brazilian foundations Coordination for the Improvement of Higher Education Personnel (CAPES - Financial Code 001), CNPq (32441/2021-2), FAPESP (2020/09838-0).

{\small\bibliography{commat}}
\EditInfo{April 1, 2023}{August 30, 2023}{Camilla Hollanti and Lenny Fukshansky}
\end{document}